\documentclass{siamonline171218}
\usepackage[T1]{fontenc}
\usepackage[latin9]{inputenc}
\usepackage{geometry}
\usepackage{xcolor}
\usepackage{pdfcolmk}
\usepackage{verbatim}
\usepackage{mathtools}

\usepackage{lipsum}
\usepackage{amsfonts}
\usepackage{graphicx}
\usepackage{epstopdf}
\usepackage{algorithmic}
\ifpdf
  \DeclareGraphicsExtensions{.eps,.pdf,.png,.jpg}
\else
  \DeclareGraphicsExtensions{.eps}
\fi

\usepackage{amsopn}

\newsiamremark{remark}{Remark}
\newsiamremark{hypothesis}{Hypothesis}
\crefname{hypothesis}{Hypothesis}{Hypotheses}
\newsiamthm{claim}{Claim}

\usepackage{pgf,tikz} 
\usepackage{pgfplots}
\pgfplotsset{compat=1.14}
\usetikzlibrary{spy}
\usetikzlibrary{calc}
\definecolor{copper}{cmyk}{0,0.9,0.9,0.2} 
\colorlet{lightgray}{black!25} 
\colorlet{darkgray}{black!75} 
\pgfmathdeclarerandomlist{MyRandomColors}{
    {blue}
    {green}    
}
\newsiamthm{condition}{Condition}
\usepackage{subcaption}

\headers{Competitive Exclusion in a DAE Model for Microbial Electrolysis Cells}{H. J. Dudley, Z. J. Ren, and D. M. Bortz}

\title{Competitive Exclusion in a DAE Model for Microbial Electrolysis Cells}

\author{Harry J. Dudley\thanks{Department of Applied Mathematics, University of Colorado, Boulder,
CO 80309-0526,\ Corresponding authors: \email{harry.dudley@colorado.edu}, \email{dmbortz@colorado.edu}}
\and Zhiyong Jason Ren\thanks{Department of Civil and Environmental Engineering, Princeton University,
Princeton, NJ 08544, USA}
\and David M. Bortz\footnotemark[1]}

\begin{document}

\maketitle

\begin{abstract}
Microbial electrolysis cells (MECs) are devices that employ electroactive
bacteria to perform extracellular electron transfer, enabling hydrogen
generation from biodegradable substrates. In our previous work, we
developed and analyzed a differential-algebraic equation (DAE) model
for MECs. The model resembles a chemostat or continuous stirred tank reactor (CSTR). 
Equations are ordinary differential equations (ODEs) for concentrations of
substrate, microorganisms, and an extracellular mediator involved in electron transfer. 
There is also an algebraic constraint for electric current and hydrogen production.
Our goal is to determine the outcome of competition between methanogenic
archaea and electroactive bacteria, because only the latter contribute
to electric current and the resulting hydrogen production. We investigate
asymptotic stability in two industrially  relevant versions of the
model. An important aspect of many chemostat models is the principle
of competitive exclusion. This states that only microbes
which grow at the lowest substrate concentration will survive
as $t\to\infty$. We show that if methanogens can grow at the lowest
substrate concentration, then the equilibrium corresponding to competitive
exclusion by methanogens is globally asymptotically stable. The analogous
result for electroactive bacteria is not necessarily true. In fact
we show that local asymptotic stability of competitive exclusion by
electroactive bacteria is not guaranteed, even in a simplified version
of the model. In this case, even if electroactive bacteria can grow
at the lowest substrate concentration, a few additional conditions
are required to guarantee local asymptotic stability. We also provide
numerical simulations supporting these arguments. Our results suggest
operating conditions that are most conducive to success of electroactive
bacteria and the resulting current and hydrogen production in MECs. This will help
identify when methane production or electricity and hydrogen production
are favored.

\end{abstract}

\maketitle

\section{Introduction}

Microbial electrolysis cells (MECs) are an emerging technology that
employs microorganisms to recover energy and resources from organic
waste \cite{lu_microbial_2016}. Bacteria on an electroactive anode
biofilm oxidize biodegradable substrate and transfer electrons, thereby
generating electrical current and releasing protons (H$^{+}$) \cite{logan_microbial_2008}.
The protons then recombine to form hydrogen at the cathode. A small
voltage (0.2--0.8 V) is needed to overcome the thermodynamic
barrier, which is much lower than traditional water electrolysis (1.8--3.5
V) and can be supplied by a small solar panel, low-grade heat, or
microbial fuel cells (MFCs), all of which can be available onsite
\cite{lu_microbial_2016,logan_microbial_2008}. The gap of energy
input between microbial and pure electrochemical electrolysis is provided
by chemical energy stored in the organics. While the electroactive
bacteria facilitate hydrogen production, methanogenic archaea consume
the same substrate to produce methane, a product which is less
energy efficient \cite{lu_active_2016}. As a result, methanogenesis leads
to decreased efficiency of the system. MEC technology has several
advantages over other resource recovery and hydrogen production methods.
Microbial electrolysis reduces energy use compared to water splitting
because some of the energy is derived from embedded energy in the
waste biomass \cite{lu_unbiased_2019}. MECs are also more efficient
than other methods using renewable wastewater, such as fermentative
hydrogen production \cite{chookaew_two-stage_2014}. In fact, \cite{lu_hydrogen_2009}
demonstrated up to 96\% recovery of the maximum theoretical yield
of hydrogen in MECs operated using fermentation effluent.

In our previous work \cite{dudley_sensitivity_2019}, we analyzed
and validated a regular, semi-explicit, index 1 differential-algebraic
equation (DAE) model for a single substrate MEC \cite{pinto2011}.
The DAE system is an extended version of an ordinary differential
equation (ODE) model for chemostats, also known as continuous stirred
tank reactors (CSTRs). Besides an ODE system that describes the rate
of change of concentrations of the microorganism populations, the
biodegradable substrate, and an extracellular mediator involved in electron
transfer, the system also includes an algebraic constraint that relates
electric current through the external circuit to the concentrations
of the electroactive bacteria and the mediator molecule. This constraint
accounts for voltage losses that occur in practice. This construct
had been used previously \cite{pinto2011} to completely avoid solving
Maxwell's equations in a partial differential equation model. It is
also commonly used in chemical fuel cell models \cite{fuel_cell_2004}.
Our group demonstrated computationally that transcritical bifurcations
in the dilution rate determine whether electroactive bacteria or methanogens
or both will survive at the stable equilibria \cite{dudley_sensitivity_2019}.
The outcome of competition for substrate is a key question because
the types of microbes that exist at the system's stable equilibria
determine the electric current and hydrogen production rate. Hydrogen
production at the stable equilibrium is possible only if electroactive
bacteria are present to generate the needed current. Our efforts here
provide answers by characterizing stability of equilibria for two
versions of the MEC model, without reference to specific parameter
values.

Our goal is to build upon extensive mathematical literature on chemostats
to characterize stability of equilibria in the MEC model. One of the
main conclusions in the chemostat literature is that if one or more
microbes can grow at a lower substrate concentration than the others,
then there is a globally asymptotically stable equilibrium where only
those microbes have nonzero concentration. This phenomenon is often
referred to as competitive exclusion and it holds under a variety
of model assumptions. Unfortunately, the MEC analysis is complicated
by the fact that growth of the mixed culture bacteria is a nonlinear
function of two interdependent variables, the concentrations of both
substrate and mediator molecules. In spite of this, we demonstrate
that competitive exclusion by methanogens is globally asymptotically
stable and provide additional conditions which are necessary for local
asymptotic stability of competitive exclusion by electroactive bacteria. The latter suggests that the conditions for competitive exclusion by electroactive bacteria are not as straightforward.

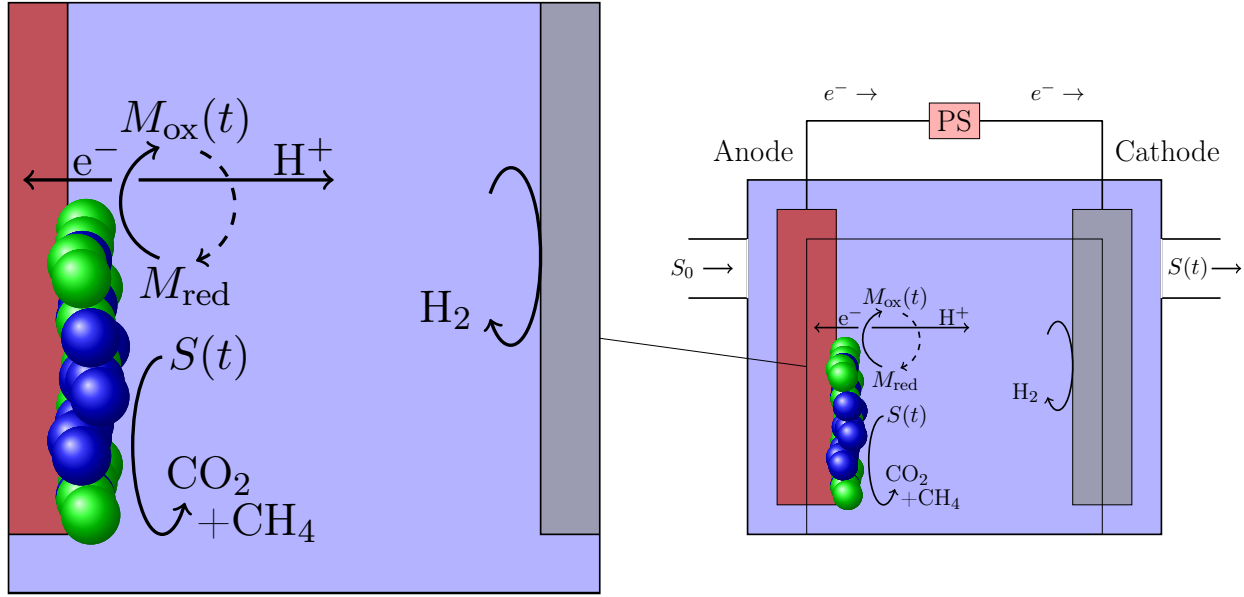
\begin{figure}
\centering{}\resizebox {\textwidth} {!} {
\begin{tikzpicture}[spy using outlines={rectangle,lens={scale=2}, size=10cm, connect spies}]

 \draw[thick,fill=blue,fill opacity = 0.3] (-1.5,1) rectangle (5.5,-5);

\draw[thick] (-2.5,0) -- (-1.5,0);     
\draw[thick] (-2.5,-1) -- (-1.5,-1);     
\draw[thick,white,opacity=1] (-1.5,0) -- (-1.5,-1);     
\draw[->, thick] (-2.25,-0.5) -- (-1.75,-0.5);     
\draw (-2.6,-0.5) node {$S_0$};          

\draw[thick] (5.5,0) -- (6.5,0);     
\draw[thick] (5.5,-1) -- (6.5,-1);     
\draw[thick,white,opacity=1] (5.5,0) -- (5.5,-1);     
\draw[->, thick] (6.35,-0.5) -- (6.85,-0.5);     
\draw (5.95,-0.5) node {$S(t)$};

\draw[fill=copper, fill opacity=0.75] (-1,0.5) rectangle (0,-4.5);     
\draw (-1.4,1.5) node {\Large Anode};          

\draw[fill=gray, fill opacity=0.75] (4,0.5) rectangle (5,-4.5);     
\draw (5.6,1.5) node {\Large Cathode};          

\draw[join = round, thick] (-0.5,0.5) -- (-0.5,2) -- (4.5,2) -- (4.5,0.5);     
\draw (2,2) node [rectangle, draw, fill=red!30] {\Large PS};    
\draw (0.25,2.5) node {$e^- \rightarrow$};     
\draw (3.75,2.5) node {$e^- \rightarrow$};  
      
\foreach \i in {0.2}{      
 \foreach \j in {-2,-2.1,...,-2.5}{      
  \pgfmathsetmacro{\dx}{rand*0.05};
  \pgfmathsetmacro{\dy}{rand*0.05};
  \pgfmathsetmacro{\rot}{rand*0.1};
  \pgfmathrandomitem{\RandomColor1}{MyRandomColors};        
  \shade[ball color=\RandomColor1] ({\i+\dx+\rot},{\j+\dy+0.4*sin(\i*3.1459265*10)}) circle(0.25cm);       
  \pgfmathrandomitem{\RandomColor2}{MyRandomColors};      
  \shade[ball color=\RandomColor2] (\i+\dx,{\j+\dy+0.4*sin(\i*3.1459265*10)-0.5}) circle(0.25cm);         
  \pgfmathrandomitem{\RandomColor3}{MyRandomColors};      
  \shade[ball color=\RandomColor3] (\i+\dx-\rot,{\j+\dy+0.4*sin(\i*3.1459265*10)-1}) circle(0.25cm); 
  \pgfmathrandomitem{\RandomColor4}{MyRandomColors};      
  \shade[ball color=\RandomColor4] (\i+\dx+\rot,{\j+\dy+0.4*sin(\i*3.1459265*10)-1.5}) circle(0.25cm);       
  \pgfmathrandomitem{\RandomColor5}{MyRandomColors};       
  \shade[ball color=\RandomColor5] (\i+\dx,{\j+\dy+0.4*sin(\i*3.1459265*10)-2}) circle(0.25cm);       
 }      
}    
 
            

\draw (1,-1) node {$M_\text{ox}(t)$};    
\draw[thick] (0.6,-1.4)+(45:.25) [yscale=2,xscale=2,<-] arc(110:250:.25);     
\draw (1,-2.4) node {$M_\text{red}$};     
\draw[thick,dashed] (1,-1.95)+(-70:.25) [yscale=2,xscale=2,<-] arc(-70:70:.25);           

\spy on (2,-2.5) in node at (-9,-1);              

\draw[->, thick] (.6,-1.5) -- (2.25,-1.5);     
\draw (2,-1.3) node {H$^+$};           

\draw[->, thick] (.375,-1.5) -- (-.375,-1.5);     
\draw (0.25,-1.3) node {e$^-$};          

\draw (1.2,-3) node {$S(t)$};     
\draw[thick] (0.8,-3) [yscale=3,->] arc(90:320:.25);     
\draw (1.2,-4) node {CO$_2$};     
\draw (1.6,-4.4) node {+CH$_4$};       

\draw[thick] (3.75,-2.5)+(-135:.25) [yscale=3,<-] arc(-135:135:.25);     
\draw (3.2,-2.6) node {H$_2$};      

\end{tikzpicture} 
}
\caption{Diagram of an MEC. An organic substrate flows into the MEC at concentration
$S_{0}$ and out at concentration $S(t)$. In the inner layer of the
anodic biofilm (highlighted by the red box), electroactive bacteria
(green spheres) oxidize the organic substrate and reduce an extracellular
mediator, $M$, thereby producing CO$_{2}$ and protons and transferring
electrons extracellularly. Methanogenic microorganisms (blue spheres)
also consume the substrate, producing CH$_{4}$ in addition to CO$_{2}$
and thereby decreasing efficiency of the system. Only methanogens
are present in the outer layer. Hydrogen is produced via a reduction
reaction as protons in solution react with electrons at the cathode.
An external voltage must be applied from a power source (labeled PS)
because the process is endothermic.\label{fig:Cartoon_MEC}}
\end{figure}

\subsection{Model\label{subsec:Model}}

Versions of the semi-explicit index 1 DAE model for the MEC have been
described previously in \cite{dudley_sensitivity_2019,pinto2011,pinto_two-population_2010}.
Our model differs from \cite{pinto2011} by not including fermenting
microorganisms that convert a complex substrate into a single compound
such as acetate. Competition among fermenting microbes is separate
from competition among electroactive bacteria and methanogens and
is not a factor in a single simple substrate MEC. Additionally, this model does not include a 
separate methanogen only biofilm layer on the anode. We consider the
following system, extended to include finitely many microbes of each
type:
\begin{align}
\dot{S}= & D\left(S_{0}-S\right) -\sum_{j=1}^{n_{m}}\frac{\mu_{m,j}(S)}{y_{m,j}}X_{m,j} -\sum_{j=1}^{n_{e}}\frac{\mu_{e,j}(S,M)}{y_{e,j}}X_{e,j},\label{eq:substrate}\\
\dot{X}_{m,j}= & \left(\mu_{m,j}\left(S\right)-K_{d,m,j}\right)X_{m,j},\ \text{for}\;j=1,2,\dots,n_{m},\label{eq:methanogen1}\\
\dot{X}_{e,j}= & \left(\mu_{e,j}\left(S,M\right)-K_{d,e,j}\right)X_{e,j},\ \text{for}\;j=1,2,\dots,n_{e},\label{eq:electroactive}\\
\dot{M}= & \gamma I-Y_{M}\sum_{j=1}^{n_{e}}\frac{\mu_{e,j}(S,M)}{y_{e,j}}X_{e,j},\label{eq:mediator}\\
0= & \Delta E-R_{\text{int}}(\mathbf{X}_{e})I-\frac{RT}{mF}\left[\ln\left(\frac{M_{\text{total}}}{M_{\text{total}}-M_{\text{}}}\right)+2\text{arcsinh}\left(\frac{I}{2A_{\text{sur,}A}I_{0}}\right)\right].\label{eq:current}
\end{align}
with initial conditions $0<S(0)\leq S_{0}$, $0<\sum_{j}X_{m,j}(0)$,
$0<\sum_{j}X_{e,j}(0)$, $0<M(0)<M_{\text{total}}$,
and $0<I(0)$. The differential equations represent concentrations,
so we are only interested in solutions with nonnegative concentrations.
We assume that initial substrate concentration, $S(0)$, is less than or equal to the influent concentration, $S_{0}$, and that initial current, $I(0)$, is positive, due to the nonzero electroactive bacteria concentration and a startup period. All of the model parameters are positive.

A single substrate with concentration $S$ flows into the tank at
constant rate $DS_{0}$, where $D=F/V$ is the flow rate per volume
and $S_{0}$ is the fixed influent substrate concentration. Let $\mathbf{X}_{m}\in\mathbb{R}^{n_{m}}$,
and $\mathbf{X}_{e}\in\mathbb{R}^{n_{e}}$
represent the concentrations of microorganisms. The anodic biofilm contains $n_{m}$
methanogen species with concentrations $X_{m,j}$ for $j=1,2,\dots,n_{m}$,
and $n_{e}$ electroactive bacteria species with concentrations $X_{e,j}$
for $j=1,2,\dots,n_{e}$. Substrate consumption is proportional
to monotonically increasing microbial growth rates, $\mu_{m,j}(S)$
or $\mu_{e,j}(S,M)$, with constants of proportionality $1/y_{m,j}$
or $1/y_{e,j}$, respectively. Each microbe also has a constant decay rate, 
$K_{d,m,j}$ or $K_{d,e,j}$. 

While methanogens only consume the substrate, electroactive bacteria
also consume the oxidized form of a mediator molecule, $M$, that
is involved in electron transfer. The mediator exists in oxidized
and reduced forms, $M_{\text{}}$ and $M_{\text{red}}$, respectively.
The mediator has a constant maximum concentration, $M_{\text{total}}=M_{\text{}}+M_{\text{red}}$,
Following \cite{pinto_two-population_2010}, electroactive
bacteria are assumed to transfer electrons via oxidation reduction
reactions of the form
\begin{gather}
S+M_{\text{}}\rightarrow M_{\text{red}}+\text{CO}_{2},\label{eq:reactions}\\
M_{\text{red}}\rightarrow M_{\text{}}+\text{e}^{-}+\text{H}^{+}.\label{eq:reactions-1}
\end{gather}
These reactions are represented in the diagram in Figure \ref{fig:Cartoon_MEC}.
The oxidized mediator is replenished at a rate proportional to the
electric current in the device, $I$, as the reduced mediator molecules
transfer electrons to the anode. $Y_{M}$ is the mediator yield of
the reactions (\ref{eq:reactions}) and (\ref{eq:reactions-1}).

The electric current is related to hydrogen production \cite{dudley_sensitivity_2019}.
This current can be determined by accounting for the voltage losses
in the system. Microbial electrolysis is endothermic, so some small
external voltage, $E_{\text{applied}}$, is required. The applied
voltage may be opposed by some counter-electromotive force, $E_{\text{CEF}}$.
There are also several sources of voltage losses that occur in practice.
Ohmic losses, $\eta_{\text{ohm}}$, arise from various types of resistance
in the circuit. Activation losses, $\eta_{\text{act}}$, arise from
the activation energy of the oxidation-reduction reactions occurring
in the cell. Concentration losses, $\eta_{conc}$, arise from certain
processes that limit the concentration of reactants at the anode and
the cathode \cite{fuel_cell_2004,logan_microbial_2006}. All of this
can be expressed in the following electrochemical balance equation,
\begin{equation}
\Delta E\coloneqq E_{\text{applied}}-E_{\text{CEF}}=\eta_{\text{ohm}}\left(\mathbf{X}_{e},I\right)+\eta_{\text{act,A}}+\eta_{\text{act,C}}\left(I\right)+\eta_{\text{conc,A}}\left(M\right)+\eta_{\text{conc,C}}\label{eq:electrochemicalBalance}
\end{equation}
where subscripts A and C represent the anode and cathode, respectively.
Previous models have ignored concentration losses at the cathode,
$\eta_{\text{conc,C}}$, because hydrogen molecules should diffuse
away from the cathode rapidly. They have also neglected activation
losses at the anode, $\eta_{\text{act,A}}$, under the assumption
that the MEC operates with higher voltage losses at the cathode. However,
these voltage losses could be included by the general constraint in
Section \ref{sec:Simplified-model}. Ohmic losses can be calculated
from Ohm's law, $\eta_{\text{ohm}}\left(\mathbf{X}_{e},I\right)=R_{\text{int}}(\mathbf{X}_{e})I$
where $R_{\text{int}}(\mathbf{X_{e})}$ is the internal resistance.
$R_{\text{int}}(\mathbf{X_{e})}$ is a decreasing function of the
total electroactive bacteria population because less electroactive
bacteria is effectively greater resistance in the circuit. Also, $R_{\text{min}}\leq R_{\text{int}}(\mathbf{X}_{e})\leq R_{\text{max}}$
with maximum resistance when there are no electroactive bacteria. Following
\cite{kato_marcus_conduction-based_2007}, concentration losses at
the anode are modeled by the Nernst equation, 
\begin{equation}
\eta_{\text{conc},A}\left(M\right)=\frac{RT}{mF}\ln\left(\frac{M_{\text{total}}}{M_{\text{total}}-M_{\text{}}}\right),\label{eq:Nernst-1}
\end{equation}
where $R$ is the ideal gas constant, $T$ is the temperature, $m$
is the number of moles of electrons transferred per mole of mediator,
and $F$ is Faraday's constant. Equation (\ref{eq:Nernst-1}) assumes
that the reference reduced mediator concentration is equal to the
total extracellular mediator concentration, $M_{\text{total}}$ \cite{pinto2011,pinto_two-population_2010}.
Activation losses at the cathode are calculated using an approximation
to the Butler-Volmer equation for the relationship between electric
current and potential at an electrode \cite{pinto2011}. Standard
simplifying assumptions are that the reaction occurs in one step and
that the symmetry coefficient (or the fraction of activation loss
that affects the rate of electrochemical transformation) is $\beta=0.5$.
With these assumptions we can write 
\begin{equation}
\eta_{\text{act},C}\left(I\right)=2\frac{RT}{mF}\text{arcsinh}\left(\frac{I}{2A_{\text{sur,}A}I_{0}}\right),\label{eq:Butler-Volmer-1}
\end{equation}
where $A_{\text{sur,}A}$ is the anode surface area and $I_{0}$ is
the reference exchange current density. See \cite{noren_clarifying_2005}
for a discussion of approximations to the Butler-Volmer equation.

The following section discusses previous work that has characterized
the equilibria of ODE systems resembling equations (\ref{eq:substrate})
- (\ref{eq:methanogen1}) when $\mathbf{X}_{e}\equiv\mathbf{0}$, $M=0$, and $I=0$. 
This paper extends some of those results to analyze
local asymptotic stability in a DAE system with additional equations
(\ref{eq:electroactive}) - (\ref{eq:current}), representing current
production by electroactive bacteria bacteria using an extracellular
mediator.

\subsection{Mathematical Background \label{subsec:Mathematical-Background}}

There is a large body of literature proving that competitive exclusion
occurs for ODE models resembling equations (\ref{eq:substrate}) -
(\ref{eq:methanogen1}) with $\mathbf{X}_{e}\equiv\mathbf{0}$, $M=0$, and $I=0$.
Essentially, stable equilibria may exhibit either
competitive exclusion (one species remains), coexistence (multiple
species remain), or total extinction (no species remain). Competitive
exclusion is generic because coexistence requires multiple species
to share an identical parameter value and extinction requires all
species to be inadequate competitors. \cite{hsu_mathematical_1977}
proved the competitive exclusion principle for an ODE model of chemostats
with microbial growth determined by Monod kinetics. \cite{hsu_limiting_1978}
then provided a more elegant proof using a Lyapunov function to guarantee
global stability. These papers showed that, if survival was possible
at all, only the microorganism(s) that could grow at the lowest substrate
concentration would survive at the stable equilibria. The work also
showed that coexistence was only possible if multiple species could
grow at the same smallest substrate concentration. \cite{hansen_single-nutrient_1980}
provided experiments verifying the theory of competitive exclusion.
Subsequently, dozens of authors have proven competitive exclusion
occurs in chemostats with various special assumptions. See the monograph
\cite{smith_theory_1995} or the more recent paper \cite{sari_global_2011}
for more details. Besides microbial growth determined by Monod kinetics,
we will also clarify our results by focusing on a simplified case of general monotonically increasing growth rates with equal washout rates and a general constraint \cite{armstrong_competitive_1980}.

To explain microbial electrolysis, the DAE model includes a differential
equation for an extracellular mediator involved in electron transfer
to the anode as well as an algebraic constraint that determines
the electric current. The constraint turns this model into a regular,
semi-explicit, index 1 DAE system which does not fit into the previous
ODE frameworks. In particular, the constraint (\ref{eq:current})
used in \cite{dudley_sensitivity_2019,pinto2011} requires a local
representation; global results may not be possible unless the electric
current constraint can be solved globally for the electrical current,
$I$. In this paper, we extend results from the chemostat literature
\cite{hsu_limiting_1978,armstrong_competitive_1980} to analyze local
asymptotic stability in the DAE system given by (\ref{eq:substrate})
- (\ref{eq:current}). The following section provides an overview
of the structure of the rest of the paper.

\subsection{Overview\label{subsec:Overview}}

Section \ref{sec:Asymptotic-stability-DAEs} reviews asymptotic stability
in semi-explicit DAEs, which is essential to the analysis in the following
sections. In particular, we review the relationship between local
asymptotic stablility and the spectrum of the matrix pencil, as well
as LaSalle's invariance principle for global stability. In Section
\ref{sec:Simplified-model}, we use the spectrum of the matrix pencil
of the DAE to characterize local asymptotic stability of equilibria
in a simplified model with 1 species of each type, general monotonically increasing
kinetics, equal decay rates, and a general constraint. This reveals
that competitive exclusion by electroactive bacteria is not locally
asymptotically stable unless the spectrum of the matrix pencil satisfies
certain conditions. Section \ref{sec:Monod-model} proves that competitive
exclusion by methanogens is globally asymptotically stable for the
full MEC system (\ref{eq:substrate}) - (\ref{eq:current}) with multiplicative
Monod kinetics, different decay rates, and a constraint based on the
Nernst and Butler-Volmer equations. However, the corresponding result for electroactive bacteria is not likely to
be true, as illustrated in the simple case in
Section \ref{sec:Simplified-model}. Numerical simulations supporting Theorem \ref{thm:competitive-exclusion-by-methanogens}
and Corollary \ref{thm:(Coexistence-of-Methanogens)} are provided
in Section \ref{sec:Numerical-Simulations}. The conclusion in Section
\ref{sec:Conclusion} summarizes results and points out that the conditions
in Section \ref{sec:Simplified-model} can be used to evaluate numerically
whether competitive exclusion by electroactive bacteria will be locally
asymptotically stable or not. The conclusion also indicates that MEC
operators will want to avoid operating conditions where methanogens
can survive at the lowest substrate value because those conditions
make competitive exclusion by methanogens globally asymptotically
stable. 

\section{Asymptotic stability in semi-explicit DAEs\label{sec:Asymptotic-stability-DAEs}}

Before presenting results on asymptotic stability of MEC equilibria
corresponding to extinction and competitive exclusion, we will briefly
review methods for determining asymptotic stability in DAEs \cite{hill_stability_1990,riaza_2008}.
The DAE framework is necessary because the constraint (\ref{eq:current})
does not admit a global solution. The MEC system in equations (\ref{eq:substrate})
- (\ref{eq:current}) can be represented as a semi-explicit DAE, 
\begin{equation}
\begin{split}\dot{x} & =f(x,y),\\
0 & =g(x,y),
\end{split}
\label{eq:semi-explicit-DAE}
\end{equation}
where $f:\mathbb{R}^{r}\times\mathbb{R}^{p}\to\mathbb{R}^{r}$ and
$g:\mathbb{R}^{r}\times\mathbb{R}^{p}\to\mathbb{R}^{p}$.
More generally, (\ref{eq:semi-explicit-DAE}) can be viewed as a quasilinear
DAE,
\begin{equation}
A(z)\dot{z}=F(z),\label{eq:quasilinear-DAE}
\end{equation}
where $A\in C^{2}(W_{0},\mathbb{R}^{r+p\times r+p})$ and $F\in C^{2}(W_{0},\mathbb{R}^{r+p})$.
Both of these perspectives will be useful. For quasilinear DAEs (\ref{eq:quasilinear-DAE}),
local asymptotic stability of equilibria can be determined from the
spectrum of the matrix pencil, $\{\sigma A(z)-F'(z):\sigma\in\mathbb{C}\}$
\cite{riaza_2008,marz_practical_1991,beardmore_stability_1998,riaza_stability_2002}.
An equilibrium point $z^{*}$ of a regular DAE is asymptotically stable
if $\text{Re}(\sigma)<0$ for all elements in $\{\sigma\in\mathbb{C}:\ \text{det\ensuremath{\left(\ensuremath{\sigma}A(z^{*})-F'(z^{*})\right)}=0\}}$
\cite{riaza_2008}. In the following section, we will use the spectrum
of the matrix pencil to analyze local asymptotic stability in a simplified
model with 1 species of each type, general monotone kinetics, equal
decay rates, and a general constraint. The analysis shows that competitive
exclusion by methanogens is locally asymptotitcally stable, but the
corresponding result is not necessarily true for competitive exclusion
by electroactive bacteria.

In the following, we assume that there is some open connected set
$\Omega\subset\mathbb{R}^{r+p}$ on which $f$ and $g$ are twice
continuously differentiable, and $g_{z}$ is nonsingular on $\Omega$.
\begin{condition}\label{DAE-requirements}  Suppose that, for some open, connected set $\Omega \in \mathbb{R}^{r+p}$, the following assumptions hold:
\begin{enumerate}
\item $f, g \in C^2(\Omega)$;
\item $g_z(y,z)$ is nonsingular on $\Omega$.
\end{enumerate}
\end{condition} 
Points where $g_{z}$ is nonsingular are called \textit{regular.}
Under the assumption of nonsingularity of $g_{z}$ on all of $W_{0}$,
(\ref{eq:semi-explicit-DAE}) is a \textit{regular DAE}, with index
1. By Condition \ref{DAE-requirements}, the implicit function
theorem allows one to describe $g(y,z)=0$ as $z=\psi(y)$ on some
open neighborhood of $y^{*}$ in $\mathbb{R}^{r}$,
where $\psi$ is twice continuously differentiable. At least locally
near $(y^{*},z^{*})$, dynamics of the differential $y$-variables
can be described by a \textit{reduced ODE, }
\begin{equation}
\dot{y}=f(y,\psi(y)).\label{eq:reduced-ODE}
\end{equation}
Solutions of (\ref{eq:reduced-ODE}) which satisfy the constraint,
$z=\psi(y)$, are solutions of the semi-explicit DAE (\ref{eq:semi-explicit-DAE})
\cite{hill_stability_1990,riaza_2008}.

We will also use the following version of LaSalle's Invariance principle
\cite{lasalle_extensions_1960} to analyze global asymptotic stability.
A well known property of regular semi-explicit, index 1 DAEs (\ref{eq:semi-explicit-DAE})
is that they define a smooth vector field on a smooth manifold \cite{riaza_2008}.
\begin{theorem}[LaSalle's invariance principle]
\label{thm:LaSalle} Consider the smooth dynamical system on an $n-$manifold given by $\dot{x} = X(x)$ and let $\Omega$ be a compact set in the manifold that is (positively) invariant under the flow of $X$. Let $V: \Omega \to \mathbb{R}$, $V \geq 0$, be a $C^1$ function such that \[   \dot{V}(x) = \frac{\partial V}{\partial x} \cdot X \leq 0 \] in $\Omega$. Let $W$ be the largest invariant set in $\Omega$ where $\dot{V}(x) = 0$. Then every solution with initial point in $\Omega$ tends asymptotically to $W$ as $t \to \infty$. In particular, if $W$ is an isolated equilibrium, it is asymptotically stable.
\end{theorem}

The function $V$ in Theorem \ref{thm:LaSalle} is called a \textit{Lyapunov
function }because LaSalle's theorem generalizes one due to Lyapunov
where $\dot{V}$ must be strictly less than zero. Section \ref{sec:Monod-model}
applies a Lyapunov function modified from \cite{hsu_limiting_1978}
to the semi-explicit DAE with finitely many species, multiplicative
Monod kinetics, different decay rates, and a constraint that is solvable
on $\Omega$.

\section{Local asymptotic stability in a simplified model \label{sec:Simplified-model}}

In this section, we consider local asymptotic stability of equilibrium
points corresponding to extinction and competitive exclusion in a
simplified MEC model. In contrast to chemostats, competitive exclusion
by electroactive bacteria is not necessarily locally asymptotically
stable, even when electroactive bacteria can grow at the lowest substrate
concentration. This is due to several issues, including the nonlinear
dependence of the growth of electroactive bacteria on both mediator
and substrate concentrations, as well as the form of the algebraic
constraint which determines the electric current. We present this
analysis to promote clarity in in Section \ref{sec:Monod-model}.

Here we simplify the model by assuming that there is one compartment for each type of microbe
(i.e., $n_{m}=n_{e}=1$), that the biofilm decay rates
are equal to the dilution rate (i.e., $K_{d,m,1}=K_{d,e,1}=D$).
We allow for general monotonically increasing kinetics
with $\mu_{m,1}(S)$ and $\mu_{e,1}(S,M)$. We also allow for a general
constraint, $0=g(X_{e,1},M,I)$, with certain reasonable derivative
conditions summarized below. These assumptions allow concise conditions
for local asymptotic stability. The presentation of results is simplified
if we rescale the model variables. We set
\[
t=\frac{\tau}{D},\quad S=sS_{0},\quad M=mY_{M},\quad X_{m,1}=x_{m}S_{0}y_{m,1}, \quad X_{e,1}=x_{e}S_{0}y_{e,1}.
\]
This yields a system of the form
\begin{align}
\dot{s}=\  & (1-s)-f_{m}\left(s\right)x_{m}-f_{e}\left(s,m\right)x_{e},\label{eq:rescaled-substrate}\\
\dot{x}_{m}=\  & \left(f_{m}\left(s\right)-1\right)x_{m},\label{eq:rescaled-methanogen-1}\\
\dot{x}_{e}=\  & \left(f_{e}\left(s,m\right)-1\right)x_{e},\label{eq:rescaled-electroactive}\\
\dot{m}=\  & \Gamma I-f_{e}\left(s,m\right)x_{e},\label{eq:rescaled-mediator}\\
0=\  & g(x_{e},m,I),\label{eq:rescaled-constraint}
\end{align}
where 
\[
f_{m}(s)\coloneqq D^{-1}\mu_{m,1}(sS_{0}),\quad f_{e}(s,m)\coloneqq D^{-1}\mu_{e,1}(sS_{0},mY_{M}),\quad\Gamma\coloneqq\gamma/Y_{M},
\]
and $g(x_{e},m,I)$ represents the rescaled constraint. Derivatives
are taken with respect to rescaled time $\tau$. Note that we can
form the new variable $u=s+x_{m}+x_{e}$ to obtain the
system
\begin{align*}
\dot{u}=\  & 1-u,\\
\dot{m}=\  & \Gamma I-S_{0}f_{e}\left(s,m\right)x_{e},\\
0=\  & g(x_{e},m,I).
\end{align*}
Since $u$ must satisfy $u(t)=1+Ce^{-t}$, solutions of $u(t)$ will
approach $1$ as $t\to\infty$. Therefore, any asymptotically stable
equilibria will satisfy $u=s+x_{m}+x_{e}=1$. We focus
on equilibria corresponding to extinction, where $s=1$, and competitive
exclusion, where either $s+x_{m}=1$ or $s+x_{e}=1$.

In this section, we will allow for general monotonically increasing
growth functions and a general constraint. We consider a class of
kinetics where the growth rates and substrate consumption rates of
the microbes will increase with substrate concentration and also mediator
concentration in the case of electroactive bacteria. We require that attainable equilibria exist. We also make
several physically reasonable assumptions based on (\ref{eq:current}).
These assumptions state that ohmic voltage losses in $g$ will decrease
with $x_{e}$ (because a decrease in electroactive bacteria concentration
is effectively an increase in resistance in the circuit), concentration
losses will increase with the oxidized mediator $m$ as the reduced
mediator becomes limited at the anode, and that activation losses
will increase with $I$. Additionally, the absence of electroactive bacteria
means that no electric current is present. Finally we require unique solutions to $g$ at two points. The requirements are summarized
by the following condition. 
\begin{condition}\label{C1}  The kinetics and constraint in this section satisfy:
\begin{enumerate}
\item $f_{m}(0) = f_{e}(0,m) = f_e(s,0) = 0$;
\item $f_{m}(s)$ and $f_{e}(s,m)$ are continuously differentiable and monotonically increasing;
\item $\exists$ $\lambda_{m}$ and $\lambda_e(m)<1$ such that $f_{m}(\lambda_{m})=1$ and $f_e(\lambda_e(m),m)=1$ for some $m \in (0,m_0]$;
\item $\frac{\partial g}{\partial x_{e}}>0$ and $\frac{\partial g}{\partial m},\frac{\partial g}{\partial I}<0$;
\item $x_e=0 \Rightarrow I=0$;
\item $\exists$ a unique positive $m_0>0$ such that $0=g(0,m_0,0)$;
\item $\exists$ a unique positive $m^*>0$ such that $0=g\left(1-\lambda_{e}(m^{*}),m^{*},\frac{1-\lambda_{e}(m^{*})}{\Gamma}\right)$.
\end{enumerate}
\end{condition}

The rescaled model (\ref{eq:rescaled-substrate}) - (\ref{eq:rescaled-constraint})
has several important equilibrium points corresponding to extinction
of all microbes or competitive exclusion by one type of microbe. These
depend on the substrate (and mediator) values where each microbe attains
zero net growth. $\lambda_{m}$ and $\lambda_{e}(m)$
are the substrate concentrations where $f_{m}\left(s\right)=1$ and $f_{e}\left(s,m\right)=1$, respectively.
In the case of the electroactive bacteria, there is a curve $s=\lambda_{e}(m)$
that satisfies $f_{e}\left(s,m\right)=1$. In Section 4, Figure \ref{fig:lambda-contours}
depicts what the curve would look like for a system with multiplicative
Monod kinetics. Biologically meaningful solutions will be located
in the interval $(0,1)$, which corresponds to attainable substrate
concentrations. If $\forall m$, $\lambda_{m},\lambda_{e}(m)\leq0$,
then no microbe can ever have positive net growth and if $\lambda_{m},\lambda_{e}(m)\geq1$,
then all microbes require more substrate than is available given the
influent substrate concentration. The model equilibria exhibiting
extinction and competitive exclusion are shown in Table \ref{tab:Equilibrium-points-rescaled}.
The mediator concentrations in these equilibria are the solutions
to the constraint, $0=\ g(x_{e},m,I)$, at the corresponding points.
In particular, $m_{0}$ is the unique positive solution to $0=g\left(0,m_{0},0\right)$
and $m^{*}$ is the unique positive solution to $0=g\left(1-\lambda_{e}(m^{*}),m^{*},\frac{1-\lambda_{e}(m^{*})}{\Gamma}\right)$.
\begin{table}[H]
\begin{centering}
\begin{tabular}{|c|c|}
\hline 
Equilibrium Point & Biological Meaning\tabularnewline
\hline 
\hline 
$p_{0}\coloneqq\left(1,0,0,m_{0},0\right)$ & Extinction of all microbes\tabularnewline
\hline 
$p_{m}\coloneqq\left(\lambda_{m},1-\lambda_{m},0,m_{0},0\right)$ & Competitive exclusion by Methanogens\tabularnewline
\hline 
$p_{e}\coloneqq\left(\lambda_{e}(m^{*}),0,1-\lambda_{e}(m^{*}),m^{*},\frac{1-\lambda_{e}(m^{*})}{\Gamma}\right)$ & Competitive exclusion by Electroactive\tabularnewline
\hline 
\end{tabular}
\par\end{centering}
\caption{Equilibrium points representing extinction and competitive exclusion
in the MEC system (\ref{eq:rescaled-substrate}) - (\ref{eq:rescaled-constraint})
with kinetics and constraints satisfying Condition \ref{C1}.
\label{tab:Equilibrium-points-rescaled}}
\end{table}

The following mutually exclusive cases make one of the equilibrium
points locally asymptotically stable. In the final case, local asymptotic
stability of the electroactive-only equilibrium also depends on a
discriminant that appears in two elements of the spectrum of the matrix
pencil at $p_{e}$; we denote this discriminant by 
\begin{equation}
\delta\coloneqq\left[\left(x_{e}^{*}\frac{\partial f_{e}}{\partial m}\frac{\partial g}{\partial I}+x_{e}^{*}\frac{\partial f_{e}}{\partial s}\frac{\partial g}{\partial I}+\Gamma\frac{\partial g}{\partial m}\right)^{2}-4\frac{\partial g}{\partial I}x_{e}^{*}\left(\Gamma\frac{\partial f_{e}}{\partial s}\frac{\partial g}{\partial m}+\Gamma\frac{\partial f_{e}}{\partial m}\frac{\partial g}{\partial x_{e}}+\frac{\partial f_{e}}{\partial m}\frac{\partial g}{\partial I}\right)\right]\big|_{p_{e}}\label{eq:discriminant}
\end{equation}
where $x_{e}^{*}=1-\lambda_{e}(m^{*})$.
\begin{enumerate}
\item (Total Extinction): If $\lambda_{m},\lambda_{e}(m_{0})\notin(0,1)$,
then $p_{0}$ is locally asymptotically stable.
\item (Competitive Exclusion by Methanogens): If $\lambda_{m}<\lambda_{e}\left(m_{0}\right)$, then $p_{m}$
is locally asymptotically stable.
\item (Competitive Exclusion by Electroactive): If $\lambda_{e}(m^{*})<\lambda_{m}$, and either (1) $\delta<0$
or (2) $\text{Re}\left(\sqrt{\delta}\right)<-\left(x_{e}^{*}\frac{\partial f_{e}}{\partial m}\frac{\partial g}{\partial I}+x_{e}^{*}\frac{\partial f_{e}}{\partial s}\frac{\partial g}{\partial I}+\Gamma\frac{\partial g}{\partial m}\right)\big|_{p_{e}}$,
then $p_{e}$ is locally asymptotically stable.
\end{enumerate}
The extinction equilibrium, $p_{0}$, will be unstable in general
because we expect at least one microbe to be an adequate competitor
with a $\lambda_\star$ value in $\left(0,1\right)$. The spectrum of the
matrix pencil at $p_{0}$ is 
\[
\left\{ -1,f_{m}\left(1\right)-1,f_{e}\left(1,m_{0}\right)-1,-\Gamma\frac{\partial g}{\partial m}/\frac{\partial g}{\partial I}\big|_{\left(0,m_{0},0\right)}\right\} ,
\]
so $p_{0}$ is unstable as long as either $\lambda_{m}$
or $\lambda_{e}(m_{0})$ are in the interval $\left(0,1\right)$.
In this case, a microbe introduced into the system may grow in the
presence of plentiful substrate. If the methanogen can grow
at the lowest substrate value, then the corresponding methanogen-only
equilibrium, $p_{m}$, will be locally asymptotically
stable. The spectrum of the matrix pencil at $p_{m}$ is
\[
\left\{ -1,\left(\lambda_{m}-1\right)f_{m}'\left(\lambda_{m}\right)-1,f_{e}\left(\lambda_{m},m_{0}\right)-1,-\Gamma\frac{\partial g}{\partial m}/\frac{\partial g}{\partial I}\big|_{\left(0,m_{0},0\right)}\right\}. 
\]
If Case 2 holds, then only $x_{m}$ will be able to attain
positive net growth near the corresponding equilibrium. Finally, the
spectrum of the matrix pencil at $p_{e}$ is $\left\{ \sigma_{1},\sigma_{2},\sigma_{3},\sigma_{4}\right\} $
where
\begin{align*}
\sigma_{1} & =\ -1,\\
\sigma_{2} & =\ f_{m}\left(\lambda_{e}(m^{*})\right)-1,\\
\sigma_{3,4} & =\ \frac{-x_{e}^{*}\frac{\partial f_{e}}{\partial m}\frac{\partial g}{\partial I}-x_{e}^{*}\frac{\partial f_{e}}{\partial s}\frac{\partial g}{\partial I}-\Gamma\frac{\partial g}{\partial m}\pm\sqrt{\delta}}{2\frac{\partial g}{\partial I}}\big|_{p_{e}}.
\end{align*}
If Case 3 holds, then $p_{e}$ is locally asymptotically stable.

Cases 1 - 3 provide conditions for local asymptotic stability of each
of the equilibria points $p_{0},\ p_{m},$ and $p_{e}$.
These apply for any continuously differentiable and monotonically
increasing growth functions $f_{m}(s)$ and $f_{e}(s,m)$
and any $g$ that satisfies Condition \ref{C1}. The conditions
in Cases 1 - 3 can be checked numerically to determine if a parametrized
model permits a locally asymptotically stable equilibrium where only
the most competitive electroactive species persists. Unlike the methanogen-only
equilibria, local asymptotic stability of the electroactive-only equilibria
is not guaranteed when the electroactive bacteria can survive at the
lowest substrate value. This means it is unlikely that chemostat results
extend to electroactive-only equilibria in an MEC. However, we can
assert that for Monod kinetics with different decay rates and a constraint
based on the Nernst and Butler-Volmer equations, methanogen-only equilibria
are globally asymptotically stable when methanogens can grow at the
lowest substrate concentration. The proof of this assertion in Section
\ref{sec:Monod-model} relies on LaSalle's invariance principle.

\section{Global asymptotic stability with Monod kinetics \label{sec:Monod-model}}

In this section, we consider the full MEC system (\ref{eq:substrate})
- (\ref{eq:current}) with multiplicative Monod kinetics, different
decay rates, and a constraint based on the Nernst and Butler-Volmer
equations. LaSalle's invariance principle allows us to show that competitive
exclusion by methanogens is globally asymptotically stable. The proof
uses a Lyapunov function adapted from \cite{hsu_limiting_1978}.
Suppose that growth rates for the methanogens are
\begin{equation}
\mu_{m,j}(S) =\mu_{\text{max},m,j}\left(\frac{S}{K_{S,m,j}+S}\right), \label{eq:methanogen-growth-1}\\
\end{equation}
and growth rates for the electroactive bacteria are 
\begin{equation}
\mu_{e,j}(S,M) =\mu_{\text{max},e,j}\left(\frac{S}{K_{S,e,j}+S}\right)\left(\frac{M}{K_{M,j}+M}\right).\label{eq:electroactive-growth}
\end{equation}
$\mu_{\text{max},m,j}$ and $\mu_{\text{max},e,j}$
are the maximum growth rates; $K_{S,m,j}$
and $K_{S,e,j}$, are half rate constants for consumption of substrate;
$K_{M,j}$ is the half rate constant for consumption of mediator.
As before, the equilibria depend on parameters or functions that are
the $S$ solutions to $\mu_{m,j}\left(S\right)=K_{d,m,j}$ and $\mu_{e,j}\left(S,M\right)=K_{d,e,j}$.
Denote these solutions by
\begin{align}
\lambda_{m,j} & \coloneqq\frac{K_{S,m,j}K_{d,m,j}}{\mu_{\text{max},m,j}-K_{d,m,j}},\label{eq:lambda-m1-def}\\
\lambda_{e,j}(M) & \coloneqq\frac{K_{S,e,j}K_{d,e,j}}{\mu_{\text{max},e,j}\left(\frac{M}{K_{M,j}+M}\right)-K_{d,e,j}}.\label{eq:lambda-e-def}
\end{align}
Assuming that the microbe concentrations are not zero,
then $\lambda_{m,j}$ and $\lambda_{e,j}(M)$ are the substrate concentrations at which
each microbe has zero net growth. The difference between the two types
of microorganisms is that each electroactive bacteria has zero net
growth on a curve $S=\lambda_{e,j}(M)$. Examples of these curves
are shown in Figure \ref{fig:lyapunov-region-1}. The fact that electroactive
bacteria have dual substrate-mediator limitation complicates the type
of analysis that has appeared in the chemostat literature.
\begin{figure}
\begin{centering}
\includegraphics{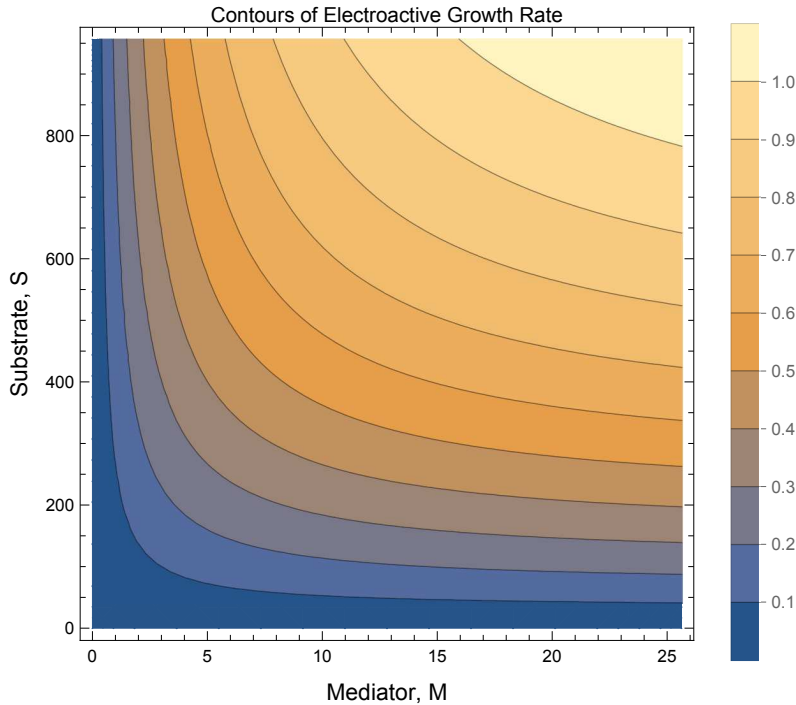}
\par\end{centering}
\caption{Contour plot of the growth rate of electroactive bacteria as a function
of mediator, $M$, and substrate, $S$, concentrations. In this contour
plot, isoclines are possible curves where an electroactive bacteria
species has zero net growth. The location of the curve $S=\lambda_{e,j}(M)$
corresponding to zero net growth is the solution to $\mu_{e,j}(S,M)=K_{d,e,j}$.
Different isoclines correspond to different values of $K_{d,e,j}$.
In contrast, the lines where methanogens have zero net growth are
horizontal in the $MS$-plane because methanogen growth does not depend
on mediator concentration. \label{fig:lambda-contours}}
\end{figure}
\begin{figure}
\centering{}%
\begin{minipage}[t]{0.45\textwidth}%
\begin{center}
\includegraphics[scale=0.75]{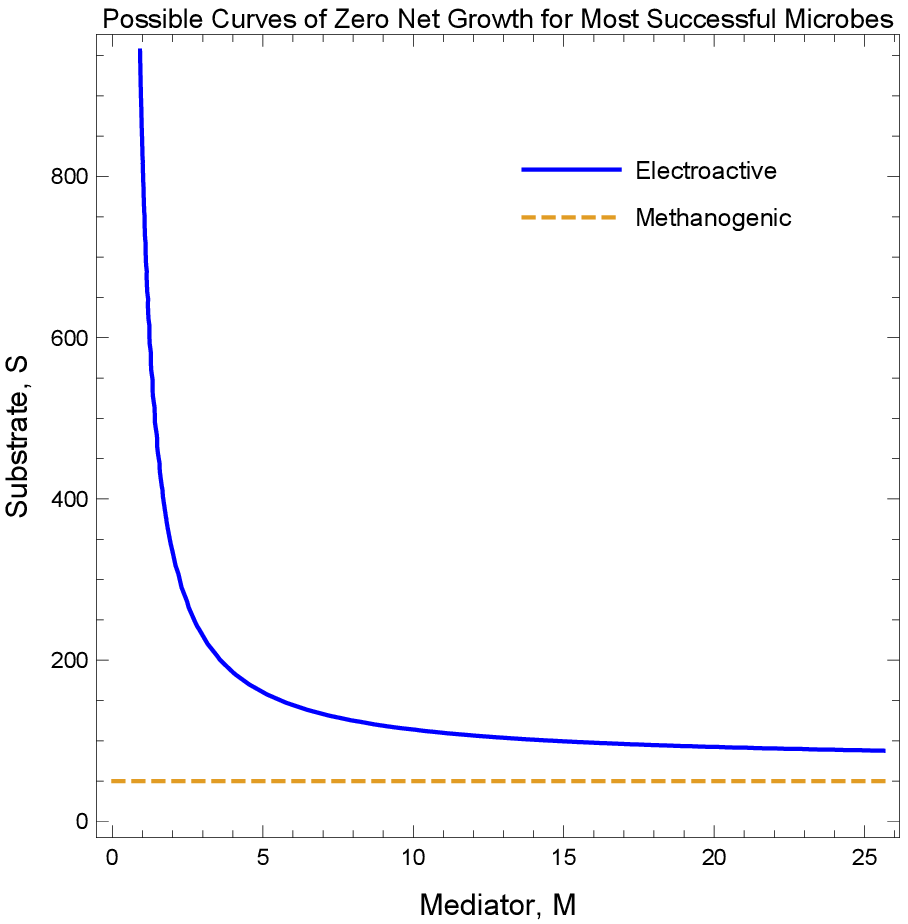}
\par\end{center}
\begin{center}
\caption{Curves along which the most successful electroactive and methanogenic
microbes have zero net growth. This figure depicts a scenario where
one of the methanogens, represented by $X_{m,1}$
survives at the lowest substrate concentration, $S=\lambda_{m,1}$,
for all obtainable mediator concentrations. In this scenario, competitive
exclusion by methanogen $X_{m,1}$ will occur. That is, the only
microbe that survives at the globally stable equilibrium will be methanogen
$X_{m,1}$. Additionally, all solutions converge to a point near
the right end of the dotted orange line in the $MS$-plane. \label{fig:lyapunov-region-1}}
\par\end{center}%
\end{minipage}\hfill{}%
\begin{minipage}[t]{0.45\textwidth}%
\begin{center}
\includegraphics[scale=0.75]{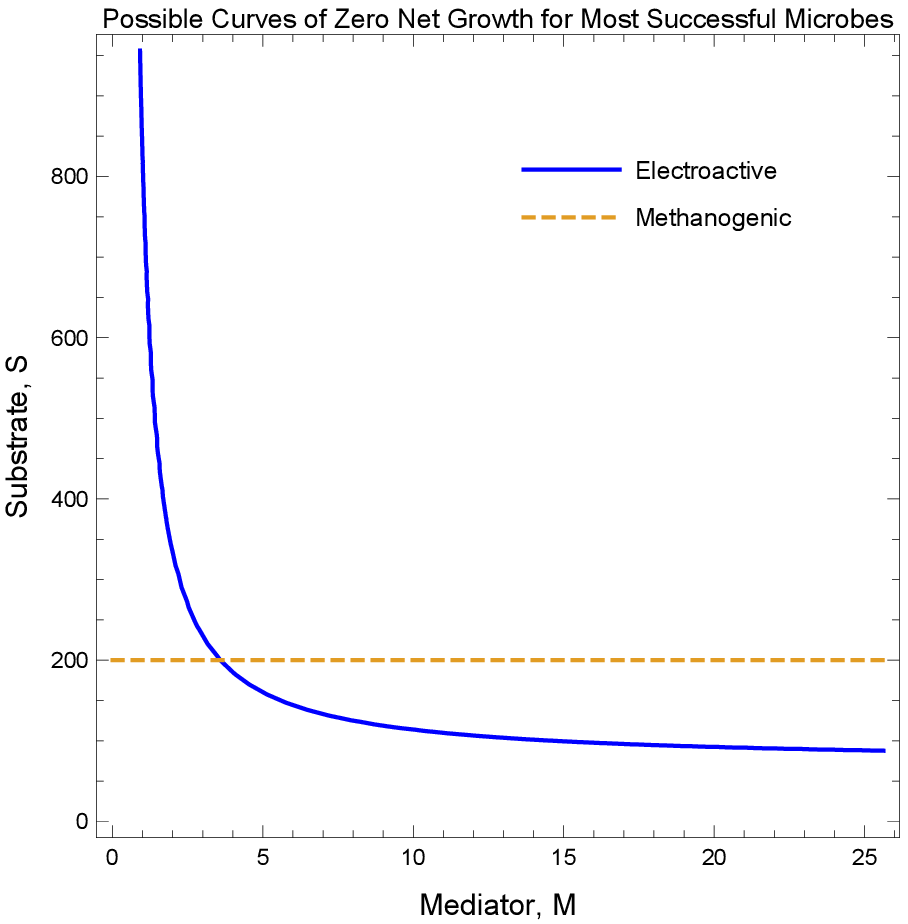}
\par\end{center}
\begin{center}
\caption{Curves along which the most successful electroactive and methanogenic
microbes have zero net growth. This figure depicts a scenario where
one of the electroactive bacteria, $X_{e,1}$, can survive at the
lowest substrate concentrations, $S=\lambda_{e,1}\left(M\right)$,
for some mediator concentrations. In this scenario, the outcome of
competition is unclear. Equilibrium points that might possibly correspond
to competitive exclusion by electroactive bacteria will be on the
blue curve below the dashed orange line. However, their location is
determined by solutions to the constraint. \label{fig:lyapunov-region-2}}
\par\end{center}%
\end{minipage}
\end{figure}

The equilibrium points corresponding to extinction and competitive
exclusion in the full unscaled model are given in Table \ref{tab:Equilibrium-points}.
Microbe concentrations at these equilibria are solutions to $\dot{S}=0$
where only one microbe survives. If $X_{m,1}$ or $X_{e,1}$ are the most competitive microbes, i.e., $\lambda_{m,1}$ or $\lambda_{e,1}(m_0)$ are the smallest $\lambda$ values, then equilibrium concentrations are
\begin{align*}
\mathbf{X}_{m}^{*}\coloneqq & \left(\frac{D(S_{0}-\lambda_{m,1})y_{m,1}}{\mu_{m,1}(\lambda_{m,1})},0,\dots,0\right),\\
\mathbf{X}_{e}^{*}\coloneqq & \left(\frac{D\left(S_{0}-\lambda_{e,1}(M^{*})\right)y_{e,1}}{\mu_{e,1}\left(\lambda_{e,1}(M^{*}),M\right)},0,\dots,0\right).
\end{align*}
The mediator concentrations $M_0$ and $M^*$ are the solutions to the constraint (\ref{eq:current})
at the corresponding points:
\begin{align*}
M_{0}\coloneqq & M_{\text{total}}\left[1-\text{exp}\left(-\frac{mF}{RT}\Delta E\right)\right],\\
M^{*}\coloneqq & M_{\text{total}}\left[1-\text{exp}\left(-\frac{mF}{RT}\left[\Delta E-2\frac{RT}{mF}\text{arcsinh}\left(\frac{I^{*}}{2I_{0}}\right)-I^{*}R_{\text{int}}(\mathbf{X}_{e}^{*})\right]\right)\right].
\end{align*}
Finally, in the case of competitive exclusion by electroactive bacteria,
the electric current is the solution to $\dot{M}=0$ when $\mathbf{X}_{e}(t)=\mathbf{X}_{e}^{*}$
and $M(t)=M^{*}$:
\[
I^{*}\coloneqq\frac{Y_{M}}{\gamma}D\left(S_{0}-\lambda_{e,1}(M^{*})\right).
\]
 
\begin{table}
\begin{centering}
\begin{tabular}{|c|c|}
\hline 
Equilibrium Point & Biological Meaning\tabularnewline
\hline 
\hline 
$P_{0}\coloneqq\left(S_{0},\mathbf{0},\mathbf{0},\mathbf{0},M_{0},0\right)$ & Extinction of all microbe species\tabularnewline
\hline 
$P_{m}\coloneqq\left(\lambda_{m,1},\mathbf{X}_{m}^{*},\mathbf{0},\mathbf{0},M_{0},0\right)$ & Competitive exclusion by methanogen $X_{m,1}$ \tabularnewline
\hline 
$P_{e}\coloneqq\left(\lambda_{e,1}(M^{*}),\mathbf{0},\mathbf{0},\mathbf{X}_{e}^{*},M^{*},I^{*}\right)$ & Competitive exclusion by electroactive bacteria $X_{e,1}$\tabularnewline
\hline 
\end{tabular}
\par\end{centering}
\caption{Equilibrium points of the MEC system (\ref{eq:substrate}) - (\ref{eq:current})
with Monod kinetics given by (\ref{eq:methanogen-growth-1}) -
(\ref{eq:electroactive-growth}). These equilibrium points represent
extinction of all microbes or competitive exclusion by one microbe
species of either type.\label{tab:Equilibrium-points}}
\end{table}

This analysis focuses on several sets of interest. Let
\begin{align}
\Omega\coloneqq \{(S,\mathbf{X}_{m},\mathbf{X}_{e},M,I): &\ 0 < S\leq S_{0},\ 0 < X_{m,j},\ 0 < X_{e,j},\label{eq:Omega}\\
 & 0< M\leq M_{0},\text{ and }0< I\leq\Delta E/R_{\text{min}}\}.\nonumber 
\end{align}
In practice, $\Omega$ is bounded because the dynamical system is dissipative, as shown in appendix \ref{subsec:Proof-of-Lemma}. The maximum concentration of each microbe must be bounded because it is not biologically possible to have infinite concentration. Although the upper bound for each microbe concentration is not clear, concentrations of each species will be bounded as $t \to \infty$. Let $\Omega_1 \subset \Omega$ be the bounded set containing these dynamics. Let $G$ be the closed set where the $C^\infty$ constraint (\ref{eq:current}) is satisfied. To obtain consistent initial conditions for the DAE, we will assume or the remainder of this section that initial conditions lie in the compact set $\Omega_G \coloneqq \Omega_1 \cap G$. 
\begin{condition}\label{C2}   
Suppose that $(S(0),\mathbf{X}_{m}(0),\mathbf{X}_{e}(0),M(0),I(0))\in\Omega_G$.
\end{condition}  Our first theorem relies on the following lemma regarding positivity and boundedness of the DAE solutions.

\begin{lemma}
\label{lem:positive-bounded} $\Omega_G$ is positively invariant for
(\ref{eq:substrate}) - (\ref{eq:current}) with Monod kinetics given
by (\ref{eq:methanogen-growth-1}) - (\ref{eq:electroactive-growth}).
\end{lemma}

We defer the proof of Lemma \ref{lem:positive-bounded} to appendix \ref{subsec:Proof-of-Lemma}. Lemma \ref{lem:positive-bounded} will be used in the proofs of the
theorem later in this section. The next lemma identifies conditions
under which a microorganism cannot survive at the stable equilibrium.
\begin{lemma}
\label{lem:Extinction} If a microbe species cannot obtain zero net
growth for $\left(S,M\right)$ values in $(0,S_{0}]\times(0,M_{0}]$,
then the concentration of that species will go to zero as $t\to\infty$.
\end{lemma}

We leave the proof in appendix \ref{subsec:Proof-of-Lemma-1}. The intuition
behind Lemma \ref{lem:Extinction} is that the substrate concentrations
where each microbe has zero net growth are $S$-coordinates of equilibria
points and they must be attainable in the interval $(0,S_{0}]$. For
electroactive bacteria, the equilibrium substrate concentrations must
be in the interval $\left(0,S_{0}\right]$ for obtainable mediator
concentrations, $M\in\left(0,M_{0}\right]$. If $\lambda_{m,j}$
and $\lambda_{e,j}(M)$ are not in this interval,
definitions (\ref{eq:lambda-m1-def}) and (\ref{eq:lambda-e-def})
tell us that either (a) the maximum growth rate is less than or equal
to the decay rate, or (b) the microbe requires more substrate than
is flowing into the device. In other words, for the microorganisms
to survive, they must be able to attain positive net growth and must
not require more substrate than is available. We will assume without
loss of generality that the following condition holds for each microbe
species; otherwise, the corresponding concentration will approach zero concentration as $t\rightarrow\infty$.
\begin{condition}\label{C3}   
Suppose that $\lambda_{m,j}\in\left(0,S_0\right]$, and ${\lambda_{e,j}(M)}\in(0,S_0] $ for some $M \in (0,M_0]$.
\end{condition}We now present the main result of this section, a theorem describing
the competitive exclusion principle in the MEC. Limiting behavior
of the DAE system (\ref{eq:substrate}) - (\ref{eq:current}) with
Monod kinetics (\ref{eq:methanogen-growth-1}) - (\ref{eq:electroactive-growth})
is determined by the smallest element in $\Lambda\coloneqq\left\{ \lambda_{m,j}\right\} {}_{j=1}^{n_{m}}\cup\left\{ \lambda_{e,j}(M_{0})\right\} _{j=1}^{n_{e}}$,
the set of smallest substrate concentrations where each microbe has
zero net growth. Intuitively, when $S$ approaches $\min\left(\Lambda\right)$
from above, all but the most competitive microbes will have negative net growth.
\begin{theorem}[Competitive Exclusion by Methanogens] \label{thm:competitive-exclusion-by-methanogens}Suppose
that Conditions \ref{C2} and \ref{C3} hold. Suppose also
that $\lambda_{m,1}$
is strictly smaller than all other elements of $\Lambda$ (i.e., methanogen
$X_{m,1}$ can survive at the lowest substrate concentration).
Then all solutions of (\ref{eq:substrate}) - (\ref{eq:current})
with Monod kinetics (\ref{eq:methanogen-growth-1}) - (\ref{eq:electroactive-growth})
will approach the point $P_{m}$ as $t\to\infty$.
\end{theorem}

\begin{proof}
(\cite{hsu_limiting_1978} provided a Lyapunov function for equations
(\ref{eq:substrate}) - (\ref{eq:methanogen1}) with $\mathbf{X}_{e}=\mathbf{0}$. That function is extended to provide these results.)
$\Omega_G\subset\mathbb{R}^{n_{m}+n_{e}+3}$ is compact
and, by Lemma \ref{lem:positive-bounded}, it is positively invariant for (\ref{eq:substrate}) - (\ref{eq:current}). Suppose that $\lambda_{m,1}$
is strictly smaller than all other elements of $\Lambda$. Let
\begin{equation}
V(S,\mathbf{X}_{m},\mathbf{X}_{m_{2}},\mathbf{X}_{e},M,I)\coloneqq \int_{\lambda_{m,1}}^{S}\frac{\sigma-\lambda_{m,1}}{\sigma}d\sigma+c_{m,1}\int_{X_{m,1}^{*}}^{X_{m,1}}\frac{\xi-X_{m,1}^{*}}{\xi}d\xi +\sum_{j=2}^{n_{m}}c_{m,j}X_{m,j}+\sum_{j=1}^{n_{e}}c_{e,j}X_{e,j}\label{eq:Lyapunov}
\end{equation}
for some unspecified constants $c_{m,j}$, $c_{m_{2},j}$, and
$c_{e,j}$. Then 
\begin{equation}
\nabla V=\left(\frac{S-\lambda_{m,1}}{S},c_{m,1}\left(\frac{X_{m,1}-X_{m,1}^{*}}{X_{m,1}}\right),c_{m,2},\dots,c_{m,n_{m}},c_{e,1},\dots,c_{e,n_{e}},0,0\right).
\end{equation}
To show that $\dot{V}\leq0$ in $\Omega_{G}$, we will use the following
rearrangements of equations (\ref{eq:methanogen1}) - (\ref{eq:electroactive}):
\begin{align}
\dot{X}_{m,j}= & \ (\mu_{\text{max},m,j}-K_{d,m,j})\left(\frac{S-\lambda_{m,j}}{K_{S,m,j}+S}\right)X_{m,j},\label{eq:methanogen1-representation}\\
\dot{X}_{e,j}= & \left(\mu_{\text{max},e,j}-K_{d,e,j}\frac{K_{M,j}+M}{M}\right)\left(\frac{S-\lambda_{e,j}(M)}{K_{S,e,j}+S}\right)\left(\frac{M}{K_{M,j}+M}\right)X_{e,j}.\label{eq:electroactive-representation}
\end{align}
Choose positive constants
\begin{align*}
c_{m,j} & =\frac{\mu_{\text{max},m,j}}{y_{m,j}(\mu_{\text{max},m,j}-K_{d,m,j})},\\
c_{e,j} & =\frac{\mu_{\text{max},e,j}}{y_{e,j}\left(\mu_{\text{max},e,j}-K_{d,e,j}\frac{K_{M,j}+M_{0}}{M_{0}}\right)}.
\end{align*}
Then
\begin{align*}
\dot{V} & =\left(\frac{S-\lambda_{m,1}}{S}\right)\left(D\left(S_{0}-S\right)-\sum_{j=1}^{n_{m}}\frac{\mu_{m,j}(S)}{y_{m,j}}X_{m,j} -\sum_{j=1}^{n_{e}}\frac{\mu_{e,j}(S,M)}{y_{e,j}}X_{e,j}\right)\\
 & \hspace{4mm}+\left(X_{m,1}-X_{m,1}^{*}\right)\left(\frac{S-\lambda_{m,1}}{S}\right)\frac{\mu_{m,1}(S)}{y_{m,1}}+\sum_{j=2}^{n_{m}}\left(\frac{S-\lambda_{m,j}}{S}\right)\frac{\mu_{m,j}(S)}{y_{m,j}}X_{m,j}\\
 & \hspace{4mm}+\sum_{j=1}^{n_{e}}c_{e,j}\left(\mu_{\text{max},e,j}-K_{d,e,j}\frac{K_{M,j}+M}{M}\right)\left(\frac{S-\lambda_{e,j}(M)}{K_{S,e,j}+S}\right)\left(\frac{M}{K_{M,j}+M}\right)X_{e,j}.
\end{align*}
Since
\[
c_{e,j}\left(\mu_{\text{max},e,j}-K_{d,e,j}\frac{K_{M,j}+M}{M}\right)=\frac{\mu_{\text{max},e,j}}{y_{e,j}}\frac{\left(\mu_{\text{max},e,j}-K_{d,e,j}\frac{K_{M,j}+M}{M}\right)}{\left(\mu_{\text{max},e,j}-K_{d,e,j}\frac{K_{M,j}+M_{0}}{M_{0}}\right)}\leq \frac{\mu_{\text{max},e,j}}{y_{e,j}},
\]
we can combine corresponding sums to obtain
\begin{align*}
\dot{V}\leq & \left(\frac{S-\lambda_{m,1}}{S}\right)\left(D\left(S_{0}-S\right)-\frac{\mu_{m,1}(S)}{y_{m,1}}X_{m,1}^{*}\right)\\
 & +\sum_{j=2}^{n_{m}}\left(\frac{\lambda_{m,1}-\lambda_{m,j}}{S}\right)\frac{\mu_{m,j}(S)}{y_{m,j}}X_{m,j}+\sum_{j=1}^{n_{e}}\left(\frac{\lambda_{m,1}-\lambda_{e,j}(M)}{S}\right)\frac{\mu_{e,j}(S,M)}{y_{e,j}}X_{e,j}.
\end{align*}
By assumption, each of the sums is less than or equal to zero. Substituting
$X_{m,1}^{*}=\frac{D\left(S_{0}-\lambda_{m,1}\right)y_{m,1}}{\mu_{m,1}\left(\lambda_{m,1}\right)}$
yields 
\[
\dot{V}\leq-\frac{D(S-\lambda_{m,1})^{2}(\lambda_{m,1}S+K_{S,m,1}S_{0})}{(K_{S,m,1}+S)S\lambda_{m,1}}\leq0.
\]
Thus, Theorem \ref{thm:LaSalle} tells us that solutions to equations
(\ref{eq:substrate}) - (\ref{eq:current}) that start in $\Omega_{G}$
approach the largest invariant set in
\begin{equation}
\begin{split}E=\{(x,y)\in\Omega_{G}:\dot{V}(x,y)=0\}=\{(\lambda_{m,1},\mathbf{X}_{m,1}^{*},\mathbf{0},M,I)\}.\end{split}
\label{eq:E-m1}
\end{equation}
From equation (\ref{eq:mediator}), we know that $\dot{M}=\gamma I\geq0$ in $E$. Thus, solutions approach
$M=M_{0}$ and $I=0$ or
\[
W=\{(\lambda_{m,1},\mathbf{X}_{m}^{*},\mathbf{0},M_{0},0)\}.
\]
\end{proof}
Theorem \ref{thm:competitive-exclusion-by-methanogens} indicates
that the methanogen that can survive at the lowest substrate concentration
will outcompete as $t\to\infty$. It is possible, albeit unlikely,
that multiple methanogen species obtain zero net growth at the same
smallest substrate concentration. The following corollary predicts
that all of these methanogens will coexist while competitively excluding
the other microbes as $t\to\infty$.
\begin{corollary}[Coexistence of Methanogens]
\label{thm:(Coexistence-of-Methanogens)}
Suppose that Conditions \ref{C2} and \ref{C3} hold and that
$\lambda=\lambda_{m,1}=\dots\lambda_{m,k}$
is strictly smaller than other distinct elements of $\Lambda$. Then
as $t\to\infty$, all solutions of (\ref{eq:substrate}) - (\ref{eq:current})
with Monod kinetics (\ref{eq:methanogen-growth-1}) - (\ref{eq:electroactive-growth})
will approach the invariant set
\begin{align*}
W=\{(\lambda,\tilde{\mathbf{X}}_{m},\mathbf{0},M_{0},0): & \ \tilde{\mathbf{X}}_{m}=\left<X_{m,1},\dots,X_{m,k},0,\dots,0\right>, D\left(S_{0}-\lambda\right)=\sum_{j=1}^{k}\frac{\mu_{m,j}(S)}{y_{m,j}}X_{m,j}\}.
\end{align*}
\end{corollary}

\begin{proof}
If $\lambda=\lambda_{m,1}=\dots\lambda_{m,k}$
are strictly smaller than all distinct elements of $\Lambda$, then
terms corresponding to $X_{m,1},\dots,X_{m,k}$ vanish from $\dot{V}$ in equation (\ref{eq:Lyapunov}) in the proof
of Theorem \ref{thm:competitive-exclusion-by-methanogens}. We still
have $\dot{V}\leq0$, but now solutions approach the largest invariant
set in $E=\{(x,y)\in\Omega_{G}:\dot{V}(x,y)=0\}=\{(\lambda,\tilde{\mathbf{X}}_{m},\mathbf{0},M,I)\}$.
Since $M \to M_0$ and $I \to 0$, solutions in $E$ approach $W=\{(\lambda,\tilde{\mathbf{X}}_{m},\mathbf{0},M_{0},0)\}$.
\end{proof}
Theorem \ref{thm:competitive-exclusion-by-methanogens} and Corollary
\ref{thm:(Coexistence-of-Methanogens)} predict that the methanogen(s)
that can grow at the smallest substrate concentration will outcompete
the other species. It is possible that multiple species will share
the same smallest $\lambda$ value, in which case coexistence of methanogens
is possible. Theorem \ref{thm:competitive-exclusion-by-methanogens}
also tells us that electroactive bacteria are guaranteed to lose the
competition if $\forall j$, $\lambda_{e,j}(M_{0})>\min\left(\Lambda\right)$. Unfortunately,
the Lyapunov function used in the proof of Theorem \ref{thm:competitive-exclusion-by-methanogens}
does not suffice to prove analogous results about global asymptotic
stability for electroactive-only equilibria. As shown in Section \ref{sec:Simplified-model},
it is likely that additional conditions on the spectrum are required
to reach conclusions about global stability of competive exclusion
by electroactive bacteria. Unfortunately, the spectrum of (\ref{eq:substrate})
- (\ref{eq:current}) cannot readily be evaluated, even when there
is only one species of each type.

The chemostat literature has shown that competitive exclusion by the
microbe that survives at the lowest substrate concentration is globally
asymptotically stable in a variety of cases. Several authors have
considered limitation by two complementary or substitutable substrates
\cite{hsu_exploitative_1981,ballyk_exploitative_1993,li_global_2000,li_how_2001},
using either monotone or minimum Monod kinetics. To our knowledge,
no one has considered multiplicative Monod kinetics of the form given
by (\ref{eq:methanogen-growth-1}) - (\ref{eq:electroactive-growth}),
particularly when one of the limiting substrates is a mediator molecule
whose concentration depends on an algebraic constraint. Figures \ref{fig:lyapunov-region-1}
and \ref{fig:lyapunov-region-2} show that the behavior of the electroactive
bacteria is more complicated than a microbe growing on one or two
substrates. One also needs detailed information about mediator concentration
because electroactive bacteria grow much more slowly when the mediator
concentration is low. Given the results in Section \ref{sec:Simplified-model},
it is not surprising that global asymptotic stability of competitive
exclusion by electroactive bacteria is more complicated. In the next
section, we provide numerical simulations supporting Theorem \ref{thm:competitive-exclusion-by-methanogens}
and Corollary \ref{thm:(Coexistence-of-Methanogens)}. These simulations
also show that competitive exclusion by electroactive bacteria may
occur if electroactive bacteria can grow at the lowest substrate concentration.

\section{Numerical Simulations \label{sec:Numerical-Simulations}}

In this section, we consider solutions of (\ref{eq:substrate}) -
(\ref{eq:current}) with Monod kinetics (\ref{eq:methanogen-growth-1})
- (\ref{eq:electroactive-growth}) and with one species of each type.
We demonstrate that when Theorem \ref{thm:competitive-exclusion-by-methanogens}
and Corollary \ref{thm:(Coexistence-of-Methanogens)} are satisfied,
then solutions behave as expected. That is, if one or more species
of methanogens can survive at the lowest substrate concentration,
then the model exhibits competitive exclusion by those methanogens.
We also demonstrate that if electroactive bacteria survive at the
lowest substrate concentration, then the model may exhibit competitive
exclusion by electroactive bacteria.

For the simulations below, we generally use parameters from the first
table of \cite{dudley_sensitivity_2019}, with the following exceptions.
The influent substrate concentration is $S_{0}=100$ and the maximum
substrate consumption rates for each species are $14$.
The maximum growth rates vary between simulations to change which
microbe can grow at the lowest substrate concentration; the parameters
are given in Table \ref{tab:Parameters-in-simulations}. Numerical
solutions were generated using the variable-order, variable-coefficient
backward differentiation formula in fixed-leading coefficient form
\cite{brenan_numerical_1995} from the IDAS package of SUNDIALS suite
of nonlinear and DAE solvers \cite{hindmarsh2005sundials}. Initial
conditions were set as $S(0)=100$, $X_{m,1}(0)=1$, $X_{m_{2},1}(0)=1$,
$M(0)=25$, and $I(0)=6$ while $X_{e,1}(0)$ is the solution to the
algebraic constraint (\ref{eq:current}).
\begin{table}[H]
\begin{centering}
\begin{tabular}{|c|c|c|c|c|}
\hline 
Parameter & Exclusion by$X_{m,1}$ & Coexistence of $X_{m,1}$, $X_{m,2}$ & Exclusion by $X_{e,1}$\tabularnewline
\hline 
\hline 
$\mu_{\text{max},e,1}$ & 1  & 1 & 5\tabularnewline
\hline 
$\mu_{\text{max},m,1}$ & 0.2  & 0.1 & 0.1\tabularnewline
\hline 
$\mu_{\text{max},m,2}$ & 0.1 & 0.1 & 0.1\tabularnewline
\hline 
$\lambda_{m,1}$ & 8.18 & 16.5 & 16.5\tabularnewline
\hline 
$\lambda_{m,2}$ & 16.5 & 16.5 & 16.5\tabularnewline
\hline 
$\lambda_{e,1}$ & 40.3 & 40.3 & 7.75\tabularnewline
\hline 
\end{tabular}
\par\end{centering}
\caption{Parameters in the three numerical simulations shown in Figure \ref{fig:Numerical-solutions}.\label{tab:Parameters-in-simulations}}
\end{table}
Figure \ref{fig:CE-m1} demonstrates that if
a methanogen can grow at the lowest substrate concentration, then
all solutions converge to a methanogen-only equilibrium, $P_{m}$. 
Figure \ref{fig:Coexistence} shows that if multiple
methanogens can survive at the same substrate concentration, then
solutions converge to a set with only those microbes. These simulations
support Theorem \ref{thm:competitive-exclusion-by-methanogens} and
Corollary \ref{thm:(Coexistence-of-Methanogens)}. Figure \ref{fig:CE-e}
shows that if an electroactive bacteria can grow at the lowest substrate
concentration, then solutions may converge to an electroactive-only
equilibrium.
\begin{figure}[H]
	\centering
	\begin{subfigure}[b]{0.475\textwidth}
		\centering
		\includegraphics[scale=0.45]{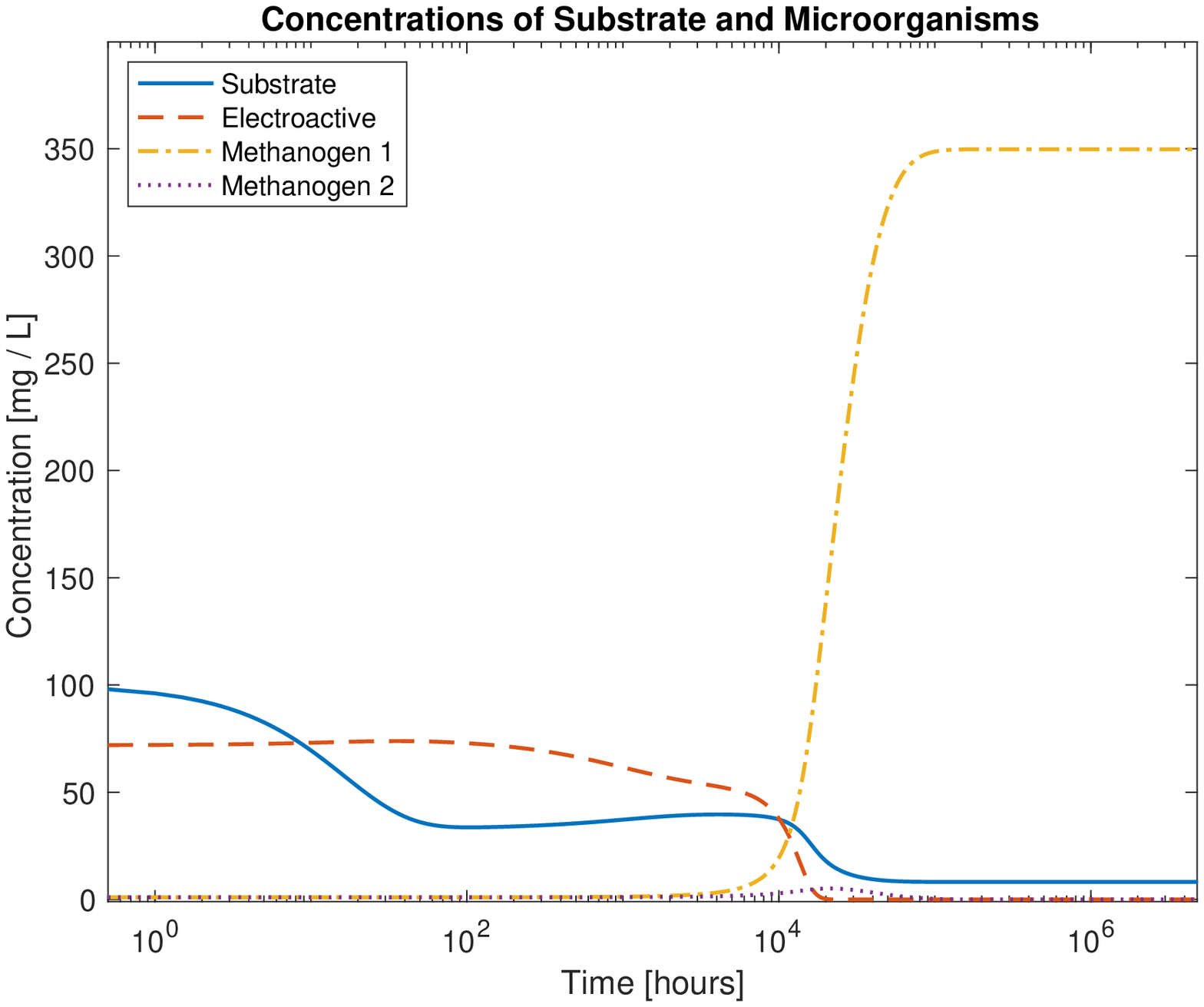}
		\caption{Competitive exclusion by methanogen 1. \label{fig:CE-m1}}
	\end{subfigure}
	\hfill	
	\begin{subfigure}[b]{0.475\textwidth}
		\centering
		\includegraphics[scale=0.45]{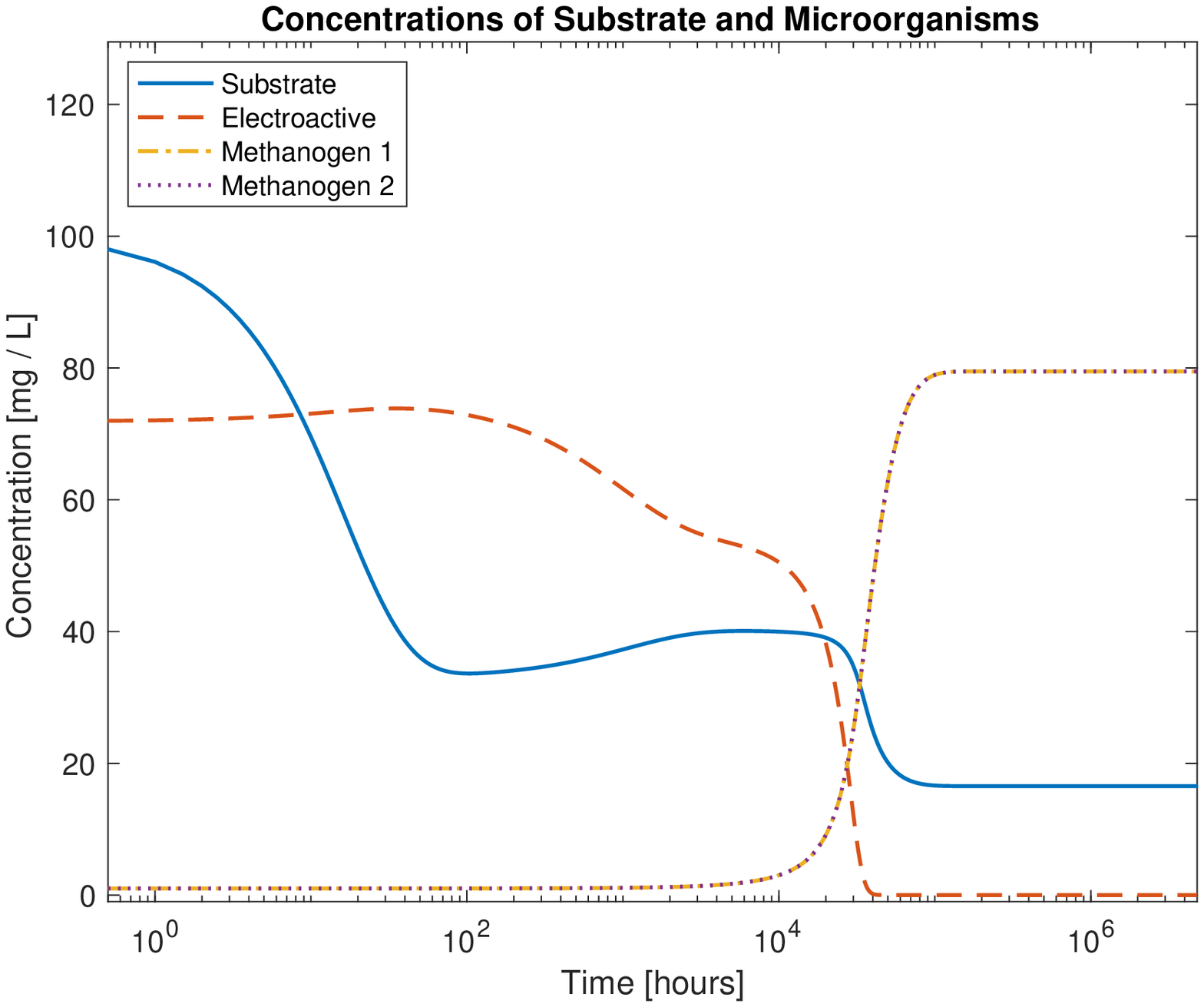}
		\caption{Coexistence of methanogens. \label{fig:Coexistence}}
	\end{subfigure}
	\vskip\baselineskip
	\begin{subfigure}[b]{0.475\textwidth}
		\centering
		\includegraphics[scale=0.45]{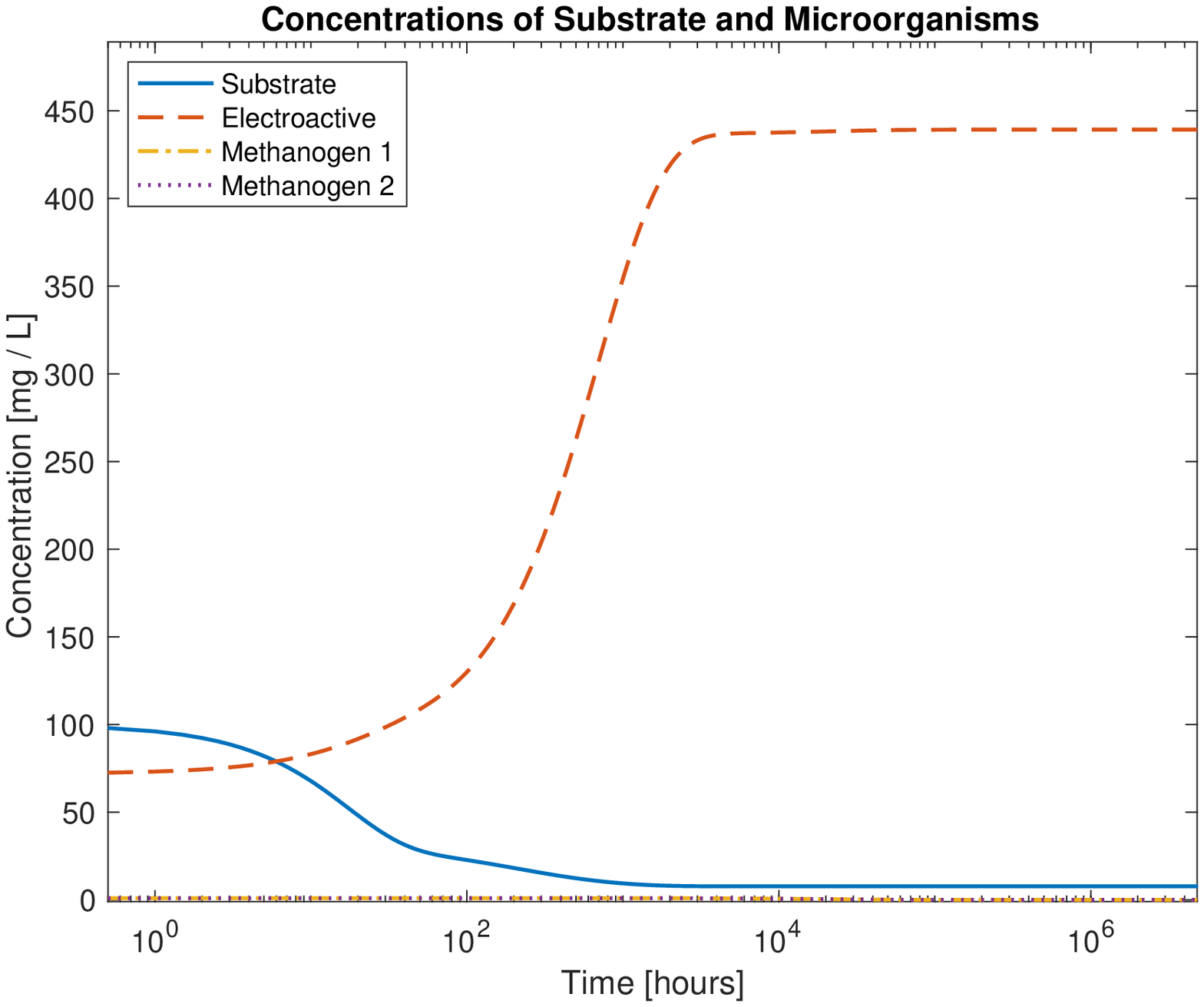}
		\caption{Competitive exclusion by electroactive bacteria. \label{fig:CE-e}}
	\end{subfigure}
	\caption{Solutions for substrate and microorganism concentration on a semi-logarithmic
plot. Substrate is depicted by a solid blue line, methanogen concentrations
are shown by a dashed-dot yellow line and a dotted purple line, respectively,
and electroactive concentration is represented as a dashed red line.
Our simulations corroborate Theorem \ref{thm:competitive-exclusion-by-methanogens}
and Corollary \ref{thm:(Coexistence-of-Methanogens)} because the
methanogen(s) that can grow at the lowest substrate concentration
are the only survivors in Subfigures (\ref{fig:CE-m1}) - (\ref{fig:Coexistence}). We have not provided
an analogous result for electroactive bacteria. Subfigure (\ref{fig:CE-e}) shows
that competitive exclusion by electroactive bacteria may occur when
these bacteria can grow at the lowest substrate concentration. However,
there may be unusual cases when this is not true, as discussed for
the simple model in Section \ref{sec:Simplified-model}. \label{fig:Numerical-solutions}}
\end{figure}

\section{Conclusion \label{sec:Conclusion}}

In Section \ref{sec:Simplified-model}, we characterized local asymptotic
stability of equilibria in a model with one species of each type,
equal decay rates, general monotone kinetics, and a general constraint.
Subsequently, in Section \ref{sec:Monod-model}, we showed that competitive
exclusion by methanogens is globally asymptotically stable in a model
with finitely many species, multiplicative Monod kinetics, different
decay rates, and a constraint based on the Nernst and Butler-Volmer
equations. Our results also show that certain operating conditions
should be avoided. In both models, if a methanogen species can grow
at the lowest substrate value, then competitive exclusion by methanogens
is either locally or globally asymptotically stable. These results
on competitive exclusion provide a recipe for MEC operation that offers
the best chance for long term electrical current and hydrogen production:
\begin{enumerate}
\item Determine which microbe can grow at the lowest substrate concentration
to ensure that methanogens will not outcompete eventually. Theorem
\ref{thm:competitive-exclusion-by-methanogens} from Section \ref{sec:Monod-model}
indicates that methanogens will competitively exclude the other microbe
species if they can survive at the lowest substrate concentration.
\item Compute the spectrum of the matrix pencil at the electroactive-only
equilibrium to ensure that it is locally asymptotically stable. Recall
from Case 3 in Section \ref{sec:Simplified-model} that if one of
the electroactive bacteria has zero net growth at the lowest substrate
concentration, $\lambda_{e,1}(m^{*})$, and that the discriminant
(\ref{eq:discriminant}) satisfies either $\delta<0$ or $\text{Re}\left(\sqrt{\delta}\right)<-\left(x_{e}^{*}\frac{\partial f_{e}}{\partial m}\frac{\partial g}{\partial I}+x_{e}^{*}\frac{\partial f_{e}}{\partial s}\frac{\partial g}{\partial I}+\Gamma\frac{\partial g}{\partial m}\right)\big|_{p_{e}}$,
then competitive exclusion by electroactive bacteria is locally asymptotically
stable. Then solutions near an electroactive-only equilibrium will
approach that equilibrium.
\end{enumerate}
In summary, if electroactive bacteria can survive at the lowest substrate
concentration and the technical condition on the discriminant of the
matrix pencil is satisfied, then electroactive bacteria are most
likely to outcompete methanogens and the MEC is most likely to provide
long term electrical current and hydrogen production. On the other
hand, if the less energy efficient methane production is desirable, one should guarantee that
a methanogen can grow at the lowest substrate concentration; then
the model will exhibit competitive exclusion by methanogens. These
results provide insight about whether the microbial electrolysis system
will produce methane or electric current and resulting hydrogen.

%
%
%
%
%
%
\appendix

\section{Proof of Lemma \ref{lem:positive-bounded} \label{subsec:Proof-of-Lemma}}
\begin{proof} 
(See \cite{wolkowicz_global_1992} for elements of the proof pertaining to $S$ and $X_{\star,j}$.) $S(t)$ is positive because $S(\tau)=0 \Rightarrow\dot{S}(\tau)=DS_{0}>0$.  For $\star \in \{m,e\},$ each $X_{\star,j}(t)$ is positive because boundaries where $X_{\star,j} = 0$ are invariant and cannot be reached in finite time if $X_{\star,j}(0)$ is positive.  $S(t)$ and $X_{\star,j}(t)$ coordinates are bounded because of the following. Define
\[
  \Sigma \coloneqq S + \sum_{j=1}^{n_{m}} \frac{X_{m,j}}{y_{m,j}} + \sum_{j=1}^{n_{e}} \frac{X_{e,j}}{y_{e,j}}.
\]
Let $\bar{D}$ be the minimum of $D$ and each $K_{d,*,j}$. Then
\[
  \dot{\Sigma}\leq DS_{0} - \bar{D}\Sigma.
\]
Thus, $S$ and each $X_{\star,j}$ is positive and bounded.  In fact, $S(t) \leq S_0$ because $S(t^*)=S_{0} \Rightarrow \dot{S}(t^*)\leq0$.

$M$ and $I$ are positive and bounded because of the algebraic constraint (\ref{eq:current}). $M$ cannot be negative because $M=0\Rightarrow\dot{M}=\gamma I\geq0$. Exponentiate both sides
of (\ref{eq:current}) and solve for $M$
to obtain 
\begin{equation}
M=M_{\text{total}}\left[1-\text{exp}\left(-\frac{mF}{RT}\left[\Delta E-2\frac{RT}{mF}\text{arcsinh}\left(\frac{I}{2 A_{\text{sur},A}I_{0}}\right)-IR_{\text{int}}(\mathbf{X}_{e})\right]\right)\right].\label{eq:M-is-bounded}
\end{equation}
Since the exponential is positive, $M<M_{\text{total}}$. In fact,
there is an even tighter bound. Define 
\[
M_{0}\coloneqq M_{\text{total}}\left[1-\text{exp}\left(-\frac{mF}{RT}\Delta E\right)\right].
\]
From equation (\ref{eq:M-is-bounded}),
\begin{align*}
I>0 & \Longleftrightarrow M<M_{0},\\
I=0 & \Longleftrightarrow M=M_{0},\\
I<0 & \Longleftrightarrow M>M_{0}.
\end{align*}
Also, equation (\ref{eq:mediator}) tells us that $I=0\Longleftrightarrow\dot{M}\leq0$.
Thus, $I(0)\geq0\Rightarrow M(0)\leq M_{0}\Rightarrow M(t)\leq M_{0}$.

$I(t)$ must be positive because $0<M(t)< M_{0}$. It remains to
show that $I(t)$ is bounded. Because $M(t)< M_{0}$ we know that
\[
1<\frac{M_{\text{total}}}{M_{\text{total}}-M(t)}\leq\frac{M_{\text{total}}}{M_{\text{total}}-M_{0}}=\text{exp}\left(\frac{mF}{RT}\Delta E\right)
\]
or
\[
0<\ln\left(\frac{M_{\text{total}}}{M_{\text{total}}-M(t)}\right)\leq\ln\left(\frac{M_{\text{total}}}{M_{\text{total}}-M_{0}}\right)=\frac{mF}{RT}\Delta E.
\]
Since $\text{arcsinh}(\frac{I}{2I_{0}})>0$ when $I>0$ and $R_{\text{min}}\leq R_{\text{int}}(\mathbf{X}_{e})\;\leq R_{\text{max}}$,
it must be true that 
\[
0=\frac{1}{R_{\text{max}}}\left[\Delta E-\frac{RT}{mF}\ln\left(\frac{M_{\text{total}}}{M_{\text{total}}-M_{0}}\right)\right]\leq I(t)\leq\frac{\Delta E}{R_{\text{min}}}.
\]
Thus, $M$ and $I$ are positive and bounded. Let $\Omega_1$ be the set where these variables are bounded. To be consistent, solutions of the DAE must lie in the closed set $G$. Thus, $\Omega_G \coloneqq \Omega_1 \cup G$ is positively invariant for (\ref{eq:substrate}) -  (\ref{eq:current}).
\end{proof}

\section{Proof of Lemma \ref{lem:Extinction}\label{subsec:Proof-of-Lemma-1}}
\begin{proof}
(See \cite{hsu_mathematical_1977} for the proof when $\mathbf{X}_{e}(t)=\mathbf{0}$
and there is no mediator or current.) If $\lambda_{m,j}\leq0$,
then $\mu_{\text{max},m,j}\leq K_{d,m,j}$. We can rearrange
the integral representation of $X_{m,j}$ and use the fact that
$S(t)\leq S_{0}$ 
to get 
\begin{align*}
X_{m,j}(t)\leq & X_{m,j}(0)\text{exp}\left(\int_{0}^{t}\frac{(\mu_{\text{max},m,j}-K_{d,m,j})S(\tau)-K_{d,m,j}K_{S,m,j}}{K_{S,m,j}+S(\tau)}d\tau\right)\\
\leq & X_{m,j}(0)\text{exp}\left(-\int_{0}^{t}\frac{K_{d,m,j}K_{S,m,j}}{K_{S,m,j}+S^{(0)}}d\tau\right)=X_{m,j}(0)\text{exp}\left(\frac{-K_{d,m,j}K_{S,m,j}}{K_{S,m,j}+S^{(0)}}t\right).
\end{align*}
$X_{m,j}(t)$ is positive and bounded by a decaying exponential,
so $\lim_{t\rightarrow\infty}X_{m,j}(t)=0$. 

If $\lambda_{m,j}>S_{0}$,
then $\mu_{\text{max},m,j}>K_{d,m,j}$. We can rearrange the
integral representation of $X_{m,j}$ and use the same fact as
before to get
\begin{align*}
X_{m,j}(t)\leq & X_{m,j}(0)\text{exp}\left(\int_{0}^{t}\left(\frac{\mu_{\text{max},m,j}-K_{d,m,j}}{K_{S,m,j}+S(\tau)}\right)\left(S(\tau)-\lambda_{m,j}\right)d\tau\right)\\
\leq & X_{m,j}(0)\text{exp}\left(\int_{0}^{t}\left(\frac{\mu_{\text{max},m,j}-K_{d,m,j}}{K_{S,m,j}+S^{(0)}}\right)\left(S^{(0)}-\lambda_{m,j}\right)d\tau\right).
\end{align*}
$X_{i,j}(t)$ is positive and bounded by a decaying exponential, so
$\lim_{t\rightarrow\infty}X_{m,j}(t)=0$.

Now consider $X_{e,j}$. If $\lambda_{e,j}(M)\leq0$ for all $M\in(0,M_{0})$,
then $\mu_{\text{max},e,j}\left(\frac{M}{K_{M,j}+M}\right)\leq K_{d,e,j}$
and $\mu_{e,j}(S,M)\leq K_{d,e,j}$ for all $S\in\left(0,S_{0}\right]$
and $M\in(0,M_{0})$. Thus,
\begin{align*}
X_{e,j}(t)\leq & X_{e,j}(0)\text{exp}\left(\int_{0}^{t}\left[\mu_{e,j}(S,M)-K_{d,e,j}\right]d\tau\right)\\
\leq & X_{e,j}(0)\text{exp}\left(-\int_{0}^{t}K_{d,e,j}d\tau\right)=X_{e,j}(0)\text{exp}\left(-K_{d,e,j}t\right).
\end{align*}
$X_{e,j}(t)$ is positive and bounded by a decaying exponential, so
$\lim_{t\rightarrow\infty}X_{e,j}(t)=0$. If instead $\lambda_{e,j}(M)>S_{0}$
for all $M\in(0,M_{0})$, then $\mu_{\text{max},e,j}\left(\frac{M}{K_{M,j}+M}\right)>K_{d,e,j}$
and $S_{0}<\lambda_{e,j}(M_{0})$. Rearrange the integral representation
to find
\begin{align*}
X_{e,j}(t)\leq & X_{e,j}(0)\text{exp}\left(\int_{0}^{t}\left[\left(\mu_{\text{max},e,j}\frac{M}{K_{M,j}+M}-K_{d,e,j}\right)\frac{\left(S-\lambda_{e,j}(M)\right)}{\left(K_{S,e,j}+S\right)}\right]d\tau\right)\\
\leq & X_{e,j}(0)\text{exp}\left(\int_{0}^{t}\left[\left(\mu_{\text{max},e,j}-K_{d,e,j}\right)\frac{\left(S_{0}-\lambda_{e,j}\left(M_{0}\right)\right)}{\left(K_{S,e,j}+S_{0}\right)}\right]d\tau\right).
\end{align*}
$X_{e,j}(t)$ is positive and bounded by a decaying exponential, so
$\lim_{t\rightarrow\infty}X_{e,j}(t)=0$.
\end{proof}

\section*{Acknowledgments}
We would like to thank Patrick De Leenheer, Jim Meiss, and Juan Restrepo for their valuable feedback and insight.

\section*{Conflict of interest}

The authors declare that they have no competing interests.

\providecommand{\href}[2]{#2}
\providecommand{\arxiv}[1]{\href{http://arxiv.org/abs/#1}{arXiv:#1}}
\providecommand{\url}[1]{\texttt{#1}}
\providecommand{\urlprefix}{URL }


\end{document}